\documentclass[onecolumn, journal]{IEEEtran}  

\IEEEoverridecommandlockouts                              
\usepackage[utf8]{inputenc}
\usepackage{amsmath,bm}
\usepackage{amsfonts}
\usepackage{amssymb}
\usepackage{graphicx}
\usepackage{amsfonts}
\usepackage{hyperref}
\usepackage{caption}
\usepackage{subcaption}
\usepackage{algorithm}
\usepackage[noend]{algpseudocode}
\usepackage{xcolor}
\usepackage{cite}
\usepackage[normalem]{ulem}

\newtheorem{Theorem}{Theorem}
\newtheorem{Proposition}{Proposition}
\newtheorem{Lemma}{Lemma}
\newtheorem{Problem}{Problem}

\newtheorem{Remark}{Remark}
\newtheorem{Assumption}{Assumption}
\usepackage{tikz}
	\usetikzlibrary{shapes,arrows}
	\tikzstyle{frame} = [draw, -latex]
	\tikzstyle{line} = [draw]
	\tikzstyle{line2} = [draw, dashdotted]
	\tikzstyle{line3} = [draw, dashed]
	\tikzstyle{line3UD} = [draw, dashed]
	\tikzstyle{place} = [circle, draw=black, fill=white, thick, inner sep=2pt, minimum size=1mm]
	\tikzstyle{place2} = [circle, draw=black, fill=black, thick, inner sep=2pt, minimum size=1mm]
	\tikzstyle{placeRed} = [circle, draw=red, fill=red, thick, inner sep=2pt, minimum size=1mm]
	\tikzstyle{vertex} = [circle, draw=black, fill=black, thick, inner sep=2pt, minimum size=1mm]
\usepackage{tkz-euclide}

\makeatletter
\def\algbackskip{\hskip-\ALG@thistlm}
\makeatother

\usepackage[switch,columnwise]{lineno}
\usepackage{setspace}

\title{\LARGE \bf Distributed traffic control for a large-scale urban network}
\author{Viet Hoang Pham$^{1}$, Kazunori Sakurama$^{2}$, Shaoshuai Mou$^{3}$ and Hyo-Sung Ahn$^{1}$ 
\thanks{This work was supported by the National Research Foundation of Korea (NRF) under the grant NRF-
2017R1A2B3007034}
\thanks{\small $^{1}$School of Mechanical Engineering, Gwangju Institute of Science and Technology, Gwangju, Korea. E-mails: {vietph@gist.ac.kr}; {hyosung@gist.ac.kr.}}
\thanks{\small $^{2}$Graduate School of Informatics, Kyoto University, Yoshida-honmachi, Sakyoku, Kyoto 606-8501 Japan. E-mail: {sakurama@i.kyoto-u.ac.jp}}
\thanks{\small $^{3}$School of Aeronautics and Astronautics, Purdue University, West Lafayette, IN 47906 USA. E-mail: {mous@purdue.edu}}
}

\begin{document}
\doublespacing
\maketitle 
\thispagestyle{empty}
\pagestyle{empty}

\begin{abstract}
Motivated by the fact that intelligent traffic control systems have become inevitable demand to cope with the risk of traffic congestion in urban areas, this paper develops a distributed control strategy for urban traffic networks.
Since these networks contain a large number of roads having different directions, each of them can be described as a multi-agent system.
Thus, a coordination among traffic flows is required to optimize the operation of the overall network.
In order to determine control decisions, we describe the objective of improving traffic conditions as a constrained optimization problem with respect to downstream traffic flows.
By applying the gradient projection method and the minimal polynomial of a matrix pair, we propose algorithms that allow each road cell to determine its control decision corresponding to the optimal solution while using only its local information.
The effectiveness of our proposed algorithms is validated by numerical simulations.
\end{abstract}
\section{Introduction} 
Traffic congestion becomes more complicated in urban areas and impacts negatively to economy, human health and environment \cite{GlenWeisbrod2001, KaiZhang2013}.
It is because urban areas are usually centers of population and economics.
The traffic demands increase dramatically while the road infrastructures are difficult or even impossible to extend. 
Dealing with this situation, the problem of urban traffic network control has attracted more and more attention from both the transportation research community and the control system society.
Since urban traffic networks usually consist of many roads and intersections, control decisions for an intersection influence to other intersections and roads.
It is also very difficult to predict future traffic behaviors because of inevitable uncertainty in the historical collected data (influenced by weather conditions, accidents, and drivers' decisions).
So, dynamical traffic control strategies are required to improve traffic conditions of the whole traffic network \cite{MarkosPapageorgiou2003, LeiChen2016}.

Among many possible methods to coordinate traffic flows through an intersection, regulating traffic lights is the most popular and easiest to implement in real practice.
In order to alleviate the risk of traffic congestion, many control approaches are proposed \cite{KonstantinosAmpountolas2009, KonstantinosAmpountolas2010, JackHaddad2010, SamuelCoogan2017, EricSKim2017, SteliosTimotheou2015, PietroGrandinetti2018}.
Their main targets are focused on optimizing some measure of interest (vehicle distribution, time delay and total throughput).
In these works, Cell Transmission Model (CTM) \cite{CarlosFDaganzo1994, CarlosFDaganzo1995} is used to establish a proper urban traffic network model since it is not only accurate in describing traffic's evolution over time and space but also simple to compute.

\begin{figure*}[t]
\begin{center}
\centering
\scalebox{0.68}{\begin{tikzpicture}[
textnode/.style={rectangle},
textlabel/.style={rectangle, scale=1.5},
squarenode/.style={rectangle, draw=black, fill=white,  very thick, minimum size=10mm},
ellipsenodeIn1/.style={ellipse, draw=black, fill=black, minimum height=4mm, minimum width= 2mm},
ellipsenodeIn2/.style={ellipse, draw=black, fill=black, minimum width=4mm, minimum height= 2mm},
ellipsenodeOut1/.style={ellipse, draw=black, fill=white, minimum height=4mm, minimum width= 2mm},
ellipsenodeOut2/.style={ellipse, draw=black, fill=white, minimum width=4mm, minimum height= 2mm},
roundnode/.style={circle, draw=black, fill=black!20,  very thick, minimum size=3mm},
]
\node at (2,5) [squarenode] (n01) {$I_1$};
\node at (5,5) [squarenode] (n02) {$I_2$};
\node at (2,2) [squarenode] (n03) {$I_3$};
\node at (5,2) [squarenode] (n04) {$I_4$};
\node at (1.7,7) [ellipsenodeIn2] (n1) {};
\node at (2.3,7) [ellipsenodeOut2] (n2) {};
\node at (4.7,7) [ellipsenodeIn2] (n3) {};
\node at (5.3,7) [ellipsenodeOut2] (n4) {};
\node at (0,5.3) [ellipsenodeOut1] (n5) {};
\node at (7,5.3) [ellipsenodeIn1] (n7) {};
\node at (0,4.7) [ellipsenodeIn1] (n8) {};
\node at (7,4.7) [ellipsenodeOut1] (n10) {};
\node at (0,2.3) [ellipsenodeOut1] (n15) {};
\node at (7,2.3) [ellipsenodeIn1] (n17) {};
\node at (0,1.7) [ellipsenodeIn1] (n18) {};
\node at (7,1.7) [ellipsenodeOut1] (n20) {};
\node at (1.7,0) [ellipsenodeOut2] (n21) {};
\node at (2.3,0) [ellipsenodeIn2] (n22) {};
\node at (4.7,0) [ellipsenodeOut2] (n23) {};
\node at (5.3,0) [ellipsenodeIn2] (n24) {};
\draw[->,{line width=3pt},black!40] (n1.south)--(n1.south|-n01.north);
\node at (1.5,6.2) [textnode] (r1) {$1$};
\draw[<-,{line width=3pt},black!40] (n2.south)--(n2.south|-n01.north);
\node at (2.5,6.2) [textnode] (r2) {$2$};
\draw[->,{line width=3pt},black!40] (n3.south)--(n3.south|-n02.north);
\node at (4.5,6.2) [textnode] (r3) {$3$};
\draw[<-,{line width=3pt},black!40] (n4.south)--(n4.south|-n02.north);
\node at (5.5,6.2) [textnode] (r4) {$4$};
\draw[<-,{line width=3pt},black!40] (n5.east)--(n5.east-|n01.west);
\node at (0.8,5.5) [textnode] (r5) {$5$};
\draw[<-,{line width=3pt},black!40, transform canvas={yshift=3mm}] (n01.east)->(n02.west);
\node at (3.5,5.5) [textnode] (r6) {$6$};
\draw[->,{line width=3pt},black!40] (n7.west)--(n7.west-|n02.east);
\node at (6.3,5.5) [textnode] (r7) {$7$};
\draw[->,{line width=3pt},black!40] (n8.east)--(n8.east-|n01.west);
\node at (0.8,4.5) [textnode] (r8) {$8$};
\draw[->,{line width=3pt},black!40, transform canvas={yshift=-3mm}] (n01.east)->(n02.west);
\node at (3.5,4.5) [textnode] (r9) {$9$};
\draw[<-,{line width=3pt},black!40] (n10.west)--(n10.west-|n02.east);
\node at (6.3,4.5) [textnode] (r10) {$10$};
\draw[<-,{line width=3pt},black!40, transform canvas={xshift=3mm}] (n01.south)->(n03.north);
\node at (1.5,3.5) [textnode] (r16) {$11$};
\draw[->,{line width=3pt},black!40, transform canvas={xshift=-3mm}] (n01.south)->(n03.north);
\node at (2.5,3.5) [textnode] (r19) {$12$};
\draw[<-,{line width=3pt},black!40, transform canvas={xshift=3mm}] (n02.south)->(n04.north);
\node at (4.5,3.5) [textnode] (r16) {$13$};
\draw[->,{line width=3pt},black!40, transform canvas={xshift=-3mm}] (n02.south)->(n04.north);
\node at (5.5,3.5) [textnode] (r19) {$14$};
\draw[<-,{line width=3pt},black!40] (n15.east)--(n15.east-|n03.west);
\node at (0.8,2.5) [textnode] (r15) {$15$};
\draw[<-,{line width=3pt},black!40, transform canvas={yshift=3mm}] (n03.east)->(n04.west);
\node at (3.5,2.5) [textnode] (r16) {$16$};
\draw[->,{line width=3pt},black!40] (n17.west)--(n17.west-|n04.east);
\node at (6.3,2.5) [textnode] (r17) {$17$};
\draw[->,{line width=3pt},black!40] (n18.east)--(n18.east-|n03.west);
\node at (0.8,1.5) [textnode] (r18) {$18$};
\draw[->,{line width=3pt},black!40, transform canvas={yshift=-3mm}] (n03.east)->(n04.west);
\node at (3.5,1.5) [textnode] (r19) {$19$};
\draw[<-,{line width=3pt},black!40] (n20.west)--(n20.west-|n04.east);
\node at (6.3,1.5) [textnode] (r10) {$20$};
\draw[<-,{line width=3pt},black!40] (n21.north)--(n21.north|-n03.south);
\node at (1.4,0.8) [textnode] (r1) {$21$};
\draw[->,{line width=3pt},black!40] (n22.north)--(n22.north|-n03.south);
\node at (2.6,0.8) [textnode] (r2) {$22$};
\draw[<-,{line width=3pt},black!40] (n23.north)--(n23.north|-n04.south);
\node at (4.4,0.8) [textnode] (r3) {$23$};
\draw[->,{line width=3pt},black!40] (n24.north)--(n24.north|-n04.south);
\node at (5.6,0.8) [textnode] (r4) {$24$};
\node at (3.5,-1.8) [textlabel] (label) {a) Graph representation of physical topology.};%

\node at (12.5,8) [roundnode] (nr01) {$1$};
\node at (14,8) [roundnode] (nr02) {$2$};
\node at (17,8) [roundnode] (nr03) {$3$};
\node at (18.5,8) [roundnode] (nr04) {$4$};
\node at (11,6.5) [roundnode] (nr05) {$5$};
\node at (15.5,6.5) [roundnode] (nr06) {$6$};
\node at (20,6.5) [roundnode] (nr07) {$7$};
\node at (11,5) [roundnode] (nr08) {$8$};
\node at (15.5,5) [roundnode] (nr09) {$9$};
\node at (20,5) [roundnode] (nr10) {$10$};
\node at (12.5,3.5) [roundnode] (nr11) {$11$};
\node at (14,3.5) [roundnode] (nr12) {$12$};
\node at (17,3.5) [roundnode] (nr13) {$13$};
\node at (18.5,3.5) [roundnode] (nr14) {$14$};
\node at (11,2) [roundnode] (nr15) {$15$};
\node at (15.5,2) [roundnode] (nr16) {$16$};
\node at (20,2) [roundnode] (nr17) {$17$};
\node at (11,0.5) [roundnode] (nr18) {$18$};
\node at (15.5,0.5) [roundnode] (nr19) {$19$};
\node at (20,0.5) [roundnode] (nr20) {$20$};
\node at (12.5,-1) [roundnode] (nr21) {$21$};
\node at (14,-1) [roundnode] (nr22) {$22$};
\node at (17,-1) [roundnode] (nr23) {$23$};
\node at (18.5,-1) [roundnode] (nr24) {$24$};
\draw[<->,{line width=1pt},black] (nr01)->(nr05);
\draw[<->,{line width=1pt},black] (nr01)->(nr06);
\draw[<->,{line width=1pt},black] (nr01)->(nr08);
\draw[<->,{line width=1pt},black] (nr01)->(nr09);
\draw[<->,{line width=1pt},black] (nr01)->(nr11);
\draw[<->,{line width=1pt},black] (nr01)->(nr12);
\draw[<->,{line width=1pt},black] (nr02)->(nr08);
\draw[<->,{line width=1pt},black] (nr02)->(nr12);
\draw[<->,{line width=1pt},black] (nr02)->(nr06);
\draw[<->,{line width=1pt},black] (nr03)->(nr06);
\draw[<->,{line width=1pt},black] (nr03)->(nr07);
\draw[<->,{line width=1pt},black] (nr03)->(nr09);
\draw[<->,{line width=1pt},black] (nr03)->(nr10);
\draw[<->,{line width=1pt},black] (nr03)->(nr13);
\draw[<->,{line width=1pt},black] (nr03)->(nr14);
\draw[<->,{line width=1pt},black] (nr04)->(nr07);
\draw[<->,{line width=1pt},black] (nr04)->(nr09);
\draw[<->,{line width=1pt},black] (nr04)->(nr14);
\draw[<->,{line width=1pt},black] (nr05)->(nr06);
\draw[<->,{line width=1pt},black] (nr05)->(nr12);
\draw[<->,{line width=1pt},black] (nr06)->(nr07);
\draw[<->,{line width=1pt},black] (nr06)->(nr08);
\draw[<->,{line width=1pt},black] (nr06)->(nr11);
\draw[<->,{line width=1pt},black] (nr06)->(nr12);
\draw[<->,{line width=1pt},black] (nr06)->(nr14);
\draw[<->,{line width=1pt},black] (nr07)->(nr09);
\draw[<->,{line width=1pt},black] (nr07)->(nr13);
\draw[<->,{line width=1pt},black] (nr07)->(nr14);
\draw[<->,{line width=1pt},black] (nr08)->(nr09);
\draw[<->,{line width=1pt},black] (nr08)->(nr11);
\draw[<->,{line width=1pt},black] (nr08)->(nr12);
\draw[<->,{line width=1pt},black] (nr09)->(nr10);
\draw[<->,{line width=1pt},black] (nr09)->(nr12);
\draw[<->,{line width=1pt},black] (nr09)->(nr13);
\draw[<->,{line width=1pt},black] (nr09)->(nr14);
\draw[<->,{line width=1pt},black] (nr10)->(nr14);
\draw[<->,{line width=1pt},black] (nr11)->(nr15);
\draw[<->,{line width=1pt},black] (nr11)->(nr16);
\draw[<->,{line width=1pt},black] (nr11)->(nr18);
\draw[<->,{line width=1pt},black] (nr11)->(nr19);
\draw[<->,{line width=1pt},black] (nr11)->(nr21);
\draw[<->,{line width=1pt},black] (nr11)->(nr22);
\draw[<->,{line width=1pt},black] (nr12)->(nr16);
\draw[<->,{line width=1pt},black] (nr12)->(nr18);
\draw[<->,{line width=1pt},black] (nr12)->(nr22);
\draw[<->,{line width=1pt},black] (nr13)->(nr16);
\draw[<->,{line width=1pt},black] (nr13)->(nr17);
\draw[<->,{line width=1pt},black] (nr13)->(nr19);
\draw[<->,{line width=1pt},black] (nr13)->(nr20);
\draw[<->,{line width=1pt},black] (nr13)->(nr23);
\draw[<->,{line width=1pt},black] (nr13)->(nr24);
\draw[<->,{line width=1pt},black] (nr14)->(nr17);
\draw[<->,{line width=1pt},black] (nr14)->(nr19);
\draw[<->,{line width=1pt},black] (nr14)->(nr24);
\draw[<->,{line width=1pt},black] (nr15)->(nr16);
\draw[<->,{line width=1pt},black] (nr15)->(nr22);
\draw[<->,{line width=1pt},black] (nr16)->(nr17);
\draw[<->,{line width=1pt},black] (nr16)->(nr18);
\draw[<->,{line width=1pt},black] (nr16)->(nr21);
\draw[<->,{line width=1pt},black] (nr16)->(nr22);
\draw[<->,{line width=1pt},black] (nr16)->(nr24);
\draw[<->,{line width=1pt},black] (nr17)->(nr19);
\draw[<->,{line width=1pt},black] (nr17)->(nr23);
\draw[<->,{line width=1pt},black] (nr17)->(nr24);
\draw[<->,{line width=1pt},black] (nr18)->(nr19);
\draw[<->,{line width=1pt},black] (nr18)->(nr21);
\draw[<->,{line width=1pt},black] (nr18)->(nr22);
\draw[<->,{line width=1pt},black] (nr19)->(nr20);
\draw[<->,{line width=1pt},black] (nr19)->(nr22);
\draw[<->,{line width=1pt},black] (nr19)->(nr23);
\draw[<->,{line width=1pt},black] (nr19)->(nr24);
\draw[<->,{line width=1pt},black] (nr20)->(nr24);
\node at (16,-1.8) [textlabel] (label) {b) Graph representation of communication topology.};
\end{tikzpicture}}
\end{center}
\caption{Graph representation for traffic network of $2\times 2$ intersections.}
\label{fig_4crossroad_graph}
\end{figure*}
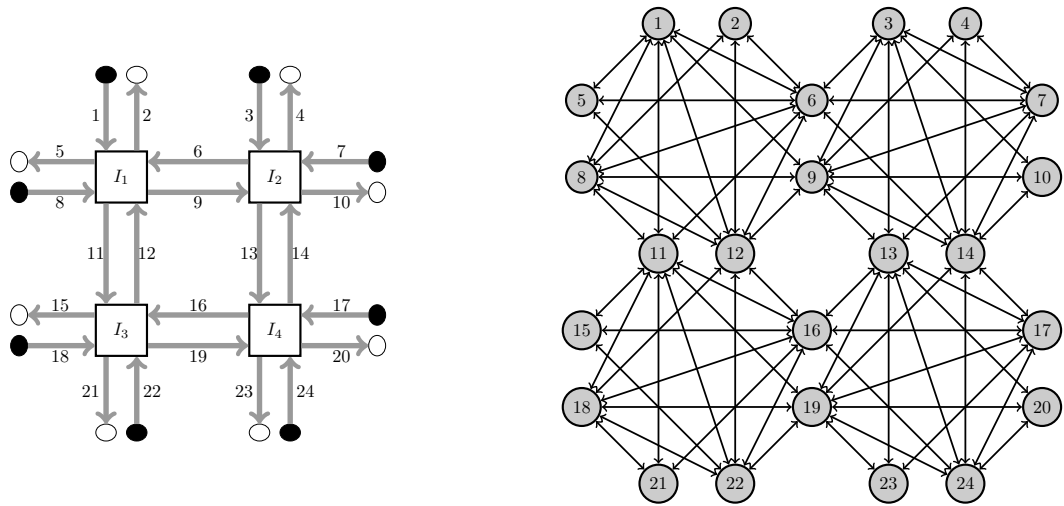

In consideration of the communication structure of controllers, traffic network management strategies are grouped into centralized \cite{KonstantinosAmpountolas2009, KonstantinosAmpountolas2010, JackHaddad2010, SamuelCoogan2017, EricSKim2017} or distributed \cite{SteliosTimotheou2015, PietroGrandinetti2018} systems.
The prominent characteristic of a centralized system is the existence of a crucial centralized controller, which gathers information of overall network to make control decisions.
On the contrary, a traffic network is controlled by the cooperation of many local controllers in the case of distributed setup.
Each local controller decides the control actions for a part of the network based on its local information.
To overcome the challenge pertaining to large-scale nature and uncertain collected data of an urban traffic network, distributed approaches are considered to be more effective than the centralized one. 
The distributed approaches have several key properties:
\begin{itemize}
\item \textit{moderate computation and information transfer load:} since local controllers require only their local information, the cost for information collecting is reduced and the risk of executing huge size data is avoided.
Moreover, the computation load of a local controller may not increase even when the size of the network increases.
\item \textit{dynamically reacting to changes:} since information transfer load is small, the delay time of information collecting is reduced significantly.
On-going traffic data affect quickly control decisions of local controllers.
Then the reliability and efficiency is increased in the distributed strategies.
\item \textit{scalable and robust:} it is not necessary to reprogram all local controllers when the traffic network structure varies.
In addition, even though a local controller fails, others are able to work independently.
\end{itemize}

In our conference version \cite{VietHoangPham2019}, we developed a distributed traffic control strategy for an urban traffic network consisting of many consecutive intersections. 
We formulated the traffic problem as a constrained Economic Dispatch Problem (EDP) and proposed a distributed algorithm to estimate an optimal distribution of vehicles for each road cell.
A consensus value computation technique proposed in \cite{ShreyasSundaram2007} was applied with the assumption that every road cell knows the minimal polynomial of the weighted matrix associated with the communication graph.
This paper is to generalize the work in \cite{VietHoangPham2019} by considering a general network and relaxing the requirement of information about overall graph.
We consider an urban traffic network whose physical graph is required to be weakly connected (i.e., its underlying graph is connected). \color{black}
To assess traffic conditions, we consider an objective function of vehicles distribution and downstream traffic flows of road cells in the considered traffic network.
The control variables we choose are downstream traffic flows since they are proportional to the active (green) time of roads and convenient for higher level controllers (for example, an intersection or an area) to synthesize their control decisions.
In order to guarantee the feasibility of control decisions, constraints for traffic volumes, traffic flows and traffic light operation are required to be satisfied.
Thus, we formulate control objective as a constrained optimization problem.  
Applying the gradient projection method and the properties of the minimal polynomial of a matrix pair, we propose a distributed online cooperative algorithm for every road cell to determine its control decision corresponding to the optimal solution for overall network.
We first provide a centralized approach solving the traffic control objective problem.
Taking the advantage of gradient projection method, we prove that the optimal solution is asymptotically achieved.
Then, we show how each road cell determines its control decision with only its local information. 
Based on minimal time consensus technique developed in \cite{ShreyasSundaram2007, YeYuan2009, YeYuan2013}, we propose an algorithm to compute the converged values of a time-invariant linear nonhomogeneous system in a minimum-time.
This minimum-time algorithm enhances a distributively implementable capacity of our proposed method.
Finally, we provide a distributed traffic control for urban networks.
\color{black}

We outline the remainder of this paper as follows. 
In Section II, we use CTM to describe urban traffic networks in multi-agent system perspective and formulate our interest as a constrained optimization problem in Section III.
By analyzing the primal problem and the dual problem, Section IV develops an iterative update for estimating the optimal solution using the gradient projection method.
We also provide a proof of convergence in this section.
In Section V, we first propose an algorithm for minimum-time final value computation and then provide a distributed traffic control for urban networks.
Section VI includes numerical simulations to validate our proposed algorithms and Section VII contains the concluding remarks.
\section{Urban traffic network model}
This paper employs Cell Transmission Model (CTM \cite{CarlosFDaganzo1994, CarlosFDaganzo1995}) to characterize traffic's evolution of road cells over time and space.
\subsection{Physical graph representation}
We use a directed graph $\mathcal{P}=(\mathcal{J},\mathcal{R})$ to illustrate the physical topology of the traffic network.
Road cells are conveniently identified with the links in the set $\mathcal{R} = \{1,2,\dots,N\}$. 
For a road cell $i \in \mathcal{R}$, we use $\sigma(i)$ and $\tau(i)$ to represent its source node and its sink node, respectively.
The direction of vehicles moving in this road cell is from $\sigma(i)$ to $\tau(i)$. 
Every node in the set $\mathcal{J} = \{I_1, \cdots, I_M\}$ corresponds to an intersection which is a cross-over of roads having different directions.
For convenience, we use $O$ to denote a virtual node representing the outside of the considered traffic network.
It includes external inputs from which the vehicles come and outputs to which the vehicles need to go.
Then $\sigma(i), \tau(i) \in \mathcal{J} \cup \{O\}$ for all $i \in \mathcal{R}$.
Let $\mathcal{N}_i^-$ and $\mathcal{N}_i^+$ denote the set of downstream neighbors and the set of upstream neighbors of road cell $i$:
\[\mathcal{N}_i^- = \{j\in\mathcal{R}:\tau(i) = \sigma(j)\},\]
\[ \textrm{ }\mathcal{N}_i^+ = \{j\in\mathcal{R}:\sigma(i) = \tau(j)\}.\] 
In other words, the set $\mathcal{N}_i^-$ consists of roads cells receiving vehicles from road cell $i$ and the set $\mathcal{N}_i^+$ includes all roads cells from which vehicles can move into the road cell $i$.
We use $\mathcal{I}_i$ to denote the set of road cells in $\mathcal{R}$ which have the same sink node as the road cell $i$ (including itself)
\[\mathcal{I}_i = \{j \in \mathcal{R}: \tau(j) = \tau(i),\}\]
i.e., road cells in $\mathcal{I}_i$ join to the same intersection.
For two road cells $j,k\in\mathcal{R}$, we say $k$ is reachable from $j$ or $j$ can reach to $k$ if there exists at least one directed sequence links $\{i_1, i_2,\dots,i_n\}$ such that $i_1=j, i_n = k$ and $\tau(i_t) = \sigma(i_{t+1}), t = 1,\dots,n-1$.
In this paper, we consider an urban traffic network whose physical topology satisfies the following assumption.
\begin{Assumption}
Any road cell $i\in\mathcal{R}$ is reachable from at least one source road cell $j\in\mathcal{R}$ where $\sigma(j) = O$; and any road cell $i\in\mathcal{R}$ can reach to at least one destination road cell $k\in\mathcal{R}$ where $\tau(k) = O$.
\label{as_graph}
\end{Assumption}

For better understanding of the physical topology representation and notations, we consider one example as shown in Fig. \ref{fig_4crossroad_graph}a.
The network consists of $24$ links and $4$ intersections, $\mathcal{R} = \{1,2,\dots,24\}$, $\mathcal{J} = \{I_1,I_2,I_3,I_4\}$.
We use gray arrows to illustrate roads in $\mathcal{R}$, use squares to represent internal intersections in $\mathcal{J}$ and use black/white circles to represent external inputs/outputs (node $O$) of the considered network.
In the figure, for example, road cells $1$ and $8$ receive exogenous inflow from outside into the network, $\sigma(1) = \sigma(8) = O$ and $\tau(1) = \tau(8) = I_1$; road cells $2$ and $5$ send vehicles to outside of the network, $\sigma(2) = \sigma(5) = I_1$ and $\tau(2) = \tau(5) = O$; road cells $6, 9, 11$ and $12$ are adjacent to intersections as $\sigma(6) = I_2$, $\sigma(12) = I_3$, $\sigma(9) = \sigma(9) = I_1$, $\tau(6) = \tau(12) = I_1$, $\tau(9) = I_2$ and $\tau(11) = I_3$.
For road cell $6$, we have $\mathcal{N}_6^- = \{2,5,11\}, \mathcal{N}_6^+ = \{3,7,14\}$ and $\mathcal{I}_6 = \{1, 6, 8, 12\}$.
\subsection{Dynamic model}
Let $t$ be the cycle index and assume all signalized intersections have a common cycle time $T$.
The discrete-time dynamical model of a road cell $i\in\mathcal{R}$ is given by the following conservation equation.
\begin{equation}
\rho_i(t+1)=\rho_i(t) + \sum\limits_{j \in \mathcal{N}_i^+}{\psi_{ji}(t)} + \mu_i(t) - \psi_i(t)
\label{eq_dynamicModel}
\end{equation}
where $\rho_i(t)$ is the traffic volume (the number of vehicles) of road cell $i$ at time $tT$, $\psi_{ji}(t)$ is the traffic flow departing from road cell $j \in \mathcal{N}_i^+$ and entering road cell $i$, $\mu_i(t)$ is the exogenous inflow entering road cell $i$ from outside of the considered network and $\psi_i(t) = \sum_{j \in \mathcal{N}_i^-}{\psi_{ij}}(t)$ is the downstream traffic flow leaving road cell $i$ in the duration time $[tT, (t+1)T]$. 

The turning preference matrix $\textbf{R}(t)=[r_{ij}(t)] \in \mathbb{R}^{N \times N}$ is defined by
\[\psi_{ij}(t) = r_{ij}(t)\psi_i(t), \textrm{ }i,j \in \mathcal{R}.\]
We have ${r_{ij}}(t) = \left\{ {\begin{array}{*{20}{c}}
{ \ge 0,}& {\tau(i) = \sigma(j)}\\{ = 0,}& {otherwise}
\end{array}} \right.\textrm{ and } \sum\limits_{k\in\mathcal{N}_i^-}{r_{ik}(t)}=1$.
\begin{Remark}
$\textbf{R}(t)$ is usually considered as known information at the beginning of each cycle of \eqref{eq_dynamicModel} (e.g., in \cite{SamuelCoogan2017, EricSKim2017, PietroGrandinetti2018}).
This assumption is acceptable since drivers usually tend to choose the lane depending on their turning choices, $r_{ij}(t)$ could be achieved empirically or by measurements.
The simplest estimation model for turning ratios is based on vehicle counts at the entrance and exit of lanes.
They can be also measured by turning ratio sensors (counting sensor, video sensors, Bluetooth and Wifi detector) installed on intersections \cite{ChangJenLan1999, MartinRodriguezVega2019, AshishBhaskar2015}.
Other strategies are the use of data recollection campaigns to construct a database of average routing behavior.
Besides that, various methods estimate the turning proportions relying on prior historical knowledge of arrival flows \cite{AfshinAbadi2015, SeunghyeonLee2015}.
\end{Remark}

Although recent advancement of sensing and information technology increases the high accuracy of collected data, it is impossible to predetermine the turning ratios and exogenous inflows exactly.
The assumption that these measurements depend on a given probability distribution is also impractical.
In this paper, we make assumption that they are given with small uncertainty.
Let $\underline{r}_{ij}(t)$ and $\overline{r}_{ij}(t)$ be the upper and lower bounds of possible values of turning split ratio $r_{ij}(t)$ and let $r_{ij}^*(t)$ be the nominal value measured or estimated by road cell $i$.
If there is no uncertainty or the measurement is accurate, $\underline{r}_{ij}(t) = \overline{r}_{ij}(t) = r_{ij}^*(t)$.
But the uncertainty cannot be avoided due to unpredictable situations and events such as weather conditions, accidents or drivers' decision change.
For the exogenous inflow $\mu_i(t)$, we define $\underline{\mu}_i(t)$, $\overline{\mu}_i(t)$ and $\mu_i^*(t)$ similarly.
\begin{Assumption}
At the beginning of every cycle $t$, each road cell $i \in \mathcal{R}$ can estimate $r_{ij}^*(t)$, $\underline{r}_{ij}(t)$, $\overline{r}_{ij}(t)$, $\mu_i^*(t)$, $\underline{\mu}_i(t)$, $\overline{\mu}_i(t)$ and
\begin{gather*}
P\left(r_{ij}(t) = r_{ij}^*(t)\right) \gg P\left(r_{ij}(t) \neq r_{ij}^*(t)\right)\\
P\left(r_{ij}(t) \in [\underline{r}_{ij}(t), \overline{r}_{ij}(t)]\right) = 1,\\
P\left(\mu_i(t) = \mu_i^*(t)\right) \gg P\left(\mu_i(t) \neq \mu_i^*(t)\right)\\
P\left(\mu_i(t) \in [\underline{\mu}_i(t), \overline{\mu}_i(t)]\right) = 1
\end{gather*}
where $P(x)$ is the the probability of $x$, $\overline{r}_{ij}(t)-\underline{r}_{ij}(t)$ and $\overline{\mu}_i(t)-\underline{\mu}_i(t)$ are small.
\label{as_parameters}
\end{Assumption}

In another words, we assume that each road cell $i$ estimates the turning split ratio of vehicles go from $i$ to $j$ as $r_{ij}^*(t)$ and its exogenous inflows as $\mu_{i}^*(t)$.
These values may be incorrect but the correct ones are in the ranges $[\underline{r}_{ij}(t), \overline{r}_{ij}(t)]$ and $[\underline{\mu}_i(t), \overline{\mu}_i(t)]$ respectively. \color{black} 
For every road cell $i$ where $\tau(i) \in \mathcal{J}$, $r_{ij}^*(t)$'s need to satisfy $\sum_{j\in\mathcal{N}_i^-}r_{ij}^*(t)=1$ but $\sum\limits_{j\in\mathcal{N}_i^-}\underline{r}_{ij}(t)<1$, $\sum\limits_{j\in\mathcal{N}_i^-}\overline{r}_{ij}(t)>1$.
\subsection{Traffic physical constraints}
To guarantee the feasibility, it is necessary to require traffic volumes, $\rho_i(t)$'s, and traffic downstream flows, $\psi_i(t)$'s traffic constraints for road capacity and traffic lights' operation.

\subsubsection{Road capacity constraints}
\begin{figure}[htb]
\begin{center}
\scalebox{0.8}{\begin{tikzpicture}
\draw[->] (0,0) -- (6,0) node[right] {$\rho$};
\draw[->] (0,0) -- (0,4.5) node[above] {};
\draw[scale=0.5,domain=0:10,smooth,variable=\x,red] plot ({\x},{8 -2/15*\x -1/15*\x*\x});
\draw[scale=0.5,domain=0:10,smooth,variable=\x,red] plot ({\x},{1.6*\x-0.08*\x*\x});
\draw[dashed] (5,0) -- (5,4);
\node at (5,-0.3) [] (ncongestion) {$\rho^{cg}$};
\draw[dashed] (2.39,0) -- (2.39,2.9);
\node at (2.39,-0.3) [] (ncritical) {$\rho^{cr}$};
\node at (4,4.2) [] (nd) {$d(\rho)$};
\node at (4,2) [] (ns) {$s(\rho)$};
\end{tikzpicture}}
\caption{The supply and demand functions, and flow capacity on a road cell $i$.}
\label{fig_diagram}
\end{center}
\end{figure}
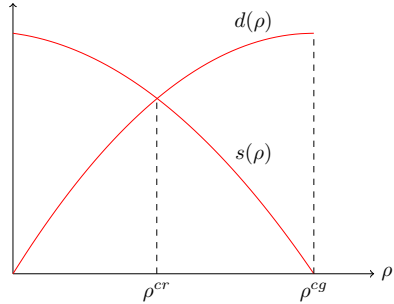
Let $d_i(\rho_i)$ be the demand function of road cell $i \in \mathcal{R}$ representing the upper bound of the downstream flow departing from $i$ when its volume is $\rho_i$, and let $s_i(\rho_i)$ be the supply function, which corresponds to the upper bound on the upstream flow entering into this road cell.
It is required that
\begin{align*}
\sum_{j\in\mathcal{N}_i^+}{\psi_{ji}(t)} + \mu_i(t) \le s_i(\rho_i(t)), \textrm{ } \psi_i(t) \le d_i(\rho_i(t)).
\end{align*}
Furthermore, the demand function $d_i$ is supposed to be non-decreasing and the supply function $s_i$ is supposed to be non-increasing (an example shown in Fig.\ref{fig_diagram}).
For each road cell $i$, there exists a congestion volume $\rho_i^{cg}$ such that $s_i(\rho_i)=0, \forall\rho_i\ge \rho_i^{cg}$ and a critical volume $\rho_i^{cr}$ where $s_i(\rho_i^{cr}) = d_i(\rho_i^{cr})$.
So, $\rho_i(t)$ needs to be smaller than its congestion volume and bigger than zero, i.e., $0 \le \rho_i(t) \le \rho_i^{cg}$ for all $t>0$.
Following Assumption \ref{as_parameters}, we have
\begin{align*}
&\sum\limits_{j\in\mathcal{N}_i^+}\underline{r}_{ji}(t){\psi_{j}(t)} \le \sum_{j\in\mathcal{N}_i^+}{\psi_{ji}(t)} \le \sum\limits_{j\in\mathcal{N}_i^+}\overline{r}_{ji}(t){\psi_{j}(t)}\\
&\rho_i(t+1) \ge \rho_i(t) + \underline{\mu}_i(t) + \sum\limits_{j\in\mathcal{N}_i^+} \underline{r}_{ji}(t)\psi_{j}(t) - \psi_{i}(t)\\
&\rho_i(t+1) \le \rho_i(t) + \overline{\mu}_i(t) + \sum\limits_{j\in\mathcal{N}_i^+} \overline{r}_{ji}(t)\psi_{j}(t) - \psi_{i}(t)
\end{align*}
Then, we need to have the physical constraint requirements as follows:
\begin{subequations}
\begin{gather}
\sum\limits_{j\in\mathcal{N}_i^+}\overline{r}_{ji}(t){\psi_{j}(t)} \le s_i(\rho_i(t)) - \overline{\mu}_i(t)\\
\rho_i(t) + \underline{\mu}_i(t) + \sum\limits_{j\in\mathcal{N}_i^+} \underline{r}_{ji}(t)\psi_{j}(t) - \psi_{i}(t) \ge 0\\
\rho_i(t) + \overline{\mu}_i(t) + \sum\limits_{j\in\mathcal{N}_i^+} \overline{r}_{ji}(t)\psi_{j}(t) - \psi_{i}(t) \le \rho_i^{cg}
\end{gather}
\label{eq_physical_constraints}
\end{subequations}

\subsubsection{Traffic light constraint}
Let $\omega_i$ be the average duration for one vehicle to exit from the road cell $i$ and enter to the intersection $\tau(i)$. 
So, for a downstream traffic flow $\psi_i(t)$, the green (active) duration $G_i(t)$ for the road cell $i$ in $t$-th cycle is required to satisfy
\[G_i(t) \ge \psi_i(t)\max_{j\in\mathcal{N}_i^-}\{r_{ij}(t)\}\omega_i.\]
Define $v_i(t) = \max_{j\in\mathcal{N}_i^-}\{r_{ij}(t)\}\omega_i$.
If $\tau(i) \in \mathcal{J}$, the set $\mathcal{I}_i$ consists of road cells which have different directions but use the same intersection.
So, the following inequality
\[\sum\limits_{j \in \mathcal{I}_i} \psi_j(t)v_j(t) \le T.\]
is required to hold since the road cells in $\mathcal{I}_i$ can not use the intersection $\tau(i)$ at the same time.
The fact $v_j(t)\psi_j(t) \le T$ implies $\psi_j(t) \le T(v_j(t))^{-1}$ for all $j$. \color{black}
\section{Problem statement} 
\subsection{Traffic objective}
The main target of large scale urban traffic network control problem is to minimize the following objective function:
\[\Phi(K) = \sum\limits_{t=1}^{K+1}{\sum\limits_{i=1}^{n}{\Phi_i(t)}}\]
where $\Phi_i(t)$ is the objective function of downstream traffic flow exiting from road cell $i$ and its traffic volume at $t$-th cycle and $K$ is the number of cycles.
From spatiotempral property of objective function, the problem of minimizing $\Phi(K)$ can be divided into $K$ sub-problems of finding $\rho_i(t)$'s and $\psi_i(t)$'s to be $\arg\min\sum_{i=1}^{N}\Phi_i(t)$ with $t = 1,\dots,K$.
The objective function $\Phi_i$ may be quadratic functions \cite{KonstantinosAmpountolas2009, KonstantinosAmpountolas2010} or have a linear form \cite{SteliosTimotheou2015}.
To have some physical meaning, $\Phi_i(t)$ could be chosen as $\left(\frac{\rho_i(t+1)}{\rho_i^{cg}}\right)^2$ to minimize the risk of congestion or as $T\rho_i(t+1)$ to minimize the total travel time.
One can also choose $\Phi_i(t) = k_i\psi_i(t)$ with negative weight $k_i$ to maximize downstream traffic flows.
Or we can use $\Phi_i(t) = \left( \rho_i(t+1) - \rho_i^{cr} \right)^2$ to achieve a given desired volume.
In this paper, we propose the following general objective function:   
\begin{equation}
\Phi_i(t) = a_i(\rho_i(t+1))^2 + b_i\rho_i(t+1) + c_i + w_i \psi_i(t)
\label{eq_cost}
\end{equation}
where $a_i, b_i, c_i, w_i$ are parameters known only by road cell $i$.

From the model \eqref{eq_dynamicModel}, we have the following matrix form equation 
\[\boldsymbol{\rho}(t+1) = \boldsymbol{\rho}(t) + \boldsymbol{\mu}(t) - \textbf{G}^T(t)\boldsymbol{\psi}(t)\]
where $\boldsymbol{\rho}(t) = [\rho_1(t),\cdots,\rho_N(t)]^T$, $\boldsymbol{\mu}(t) = [\mu_1(t),\cdots,\mu_N(t)]^T$, $\boldsymbol{\psi}(t) = [\psi_1(t),\cdots,\psi_N(t)]^T$ and the matrix $\textbf{G}(t) = [g_{ij}(t)] \in \mathbb{R}^{N \times N}$ is defined by
\begin{equation}
g_{ij}(t) = \left\{ {\begin{array}{*{20}{cl}} 1 & j=i\\ -r_{ij}(t) & j\in\mathcal{N}_i^-\\ 0 & otherwise \end{array}} \right.
\label{eq_coeff_vector}
\end{equation}
We have $\textbf{G}(t) = \textbf{I}_N - \textbf{R}(t)$ and it is a positive definite matrix because of Lemma \ref{lemma_turning_matrix}.
\begin{Lemma} 
The matrix $\textbf{R}(t)$ has the spectral radius less than $1$.
\label{lemma_turning_matrix}
\end{Lemma}
\begin{proof} 
According to the Gerschgorin circles theorem, we have all eigenvalues of $\textbf{R}(t)$ lie within the union of the $N$ circles defined by
\[|z - r_{ii}(t)| \le - |r_{ii}(t)| + \sum\limits_{j = 1}^{N} |r_{ij}(t)|, i = 1, \dots, N\]
From the definition of the turning preference matrix, we have $r_{ii}(t) = 0$ for all $i \in \mathcal{R}$ and the sum $\sum_{j = 1}^{N} |r_{ij}(t)|$ equals to one if $\tau(i)\in\mathcal{J}$ or to zero if $\tau(i) = O$.
This means \[|\kappa| \le 1,\]
for every eigenvalue $\kappa$ of the matrix $\textbf{R}(t)$.
Let $\textbf{u} = \left[ {\begin{array}{*{20}{cl}} u_1 & u_2 & \dots & u_N \end{array}} \right]^T$ be the eigenvector corresponding to eigenvalue $\kappa$, we have 
\[r_{i1}(t)u_1 + r_{i2}(t)u_2 + \dots + r_{iN}(t)u_N = \kappa u_i\]
for all $i \in \mathcal{R}$.
Consider the case of $|\kappa| = 1$, then
\begin{align*}
|\kappa u_i| = |u_i| &= |r_{i1}(t)u_1 + r_{i2}(t)u_2 + \dots + r_{iN}(t)u_N|\\
&\le |r_{i1}(t)||u_1| + |r_{i2}(t)||u_2| + \dots + |r_{iN}(t)||u_N|
\end{align*}
Since $r_{ij}(t) \ge 0$ for all $i,j \in \mathcal{R}$, we have
\begin{equation}
|u_i| \le r_{i1}(t)|u_1| + r_{i2}(t)|u_2| + \dots + r_{iN}(t)|u_N|
\label{eq_absInequality}
\end{equation}
Let $i^*$ be the index where $|u_{i^*}| = \max\left\{ |u_j|, j \in \mathcal{R}\right\}$ and $|u_{i^*}| > 0$.
We have
\[|u_{i^*}| = r_{{i^*}1}(t)|u_{i^*}| + r_{{i^*}2}(t)|u_{i^*}| + \dots + r_{{i^*}N}(t)|u_{i^*}|.\]
because $\sum_{j=1}^{N} r_{i^*j}(t) = 1$.
It implies $\sum_{j=1}^{N} r_{i^*j}(t)(|u_j| - |u_{i^*}|) \le 0$.
Because $r_{i^*j}(t) > 0$ if $j \in \mathcal{N}_{i^*}^-$ and $r_{i^*j}(t) = 0$ if $j \notin \mathcal{N}_{i^*}^-$, the inequality \eqref{eq_absInequality} holds for $i^*$ if and only if $|u_j| = |u_{i^*}|$ for all $j \in \mathcal{N}_{i^*}^-$.
Then we generalize that for all $j$ which is reachable from $i^*$, we have $|u_j| = |u_{i^*}|$.

By Assumption \ref{as_graph}, there exists $j^{*}$ where $\tau(j^*) = O$ such that $|u_{j^*}| = |u_{i^*}|$.
Since $r_{j^*k}(t) = 0$ for all $k \in \mathcal{R}$, \eqref{eq_absInequality} implies $|u_{j^*}| \le 0$ or $u_j^* = 0$.
That means $\max\left\{ |u_i|, i \in \mathcal{R}\right\} = 0$ or $u_i = 0$ for all $i = 1, \dots, N$.
This contradicts with the assumption $\textbf{u}$ is an eigenvector of $\textbf{R}(t)$.

Thus, $|\kappa| < 1$ for all eigenvalues of $\textbf{R}(t)$.
\end{proof}
Since $\textbf{G}(t)$ is a nonsingular matrix, the desired traffic volumes of road cells can be reached by regulating their downstream traffic flows.
\subsection{Problem statement} 
Denoting $x_i = \rho_i(t+1)$, $f_i = \psi_i(t)$ and $x_i^0 = \rho_i(t)+\mu_i(t)$, the objective function at $t$-th cycle is given by
\begin{equation}
\Phi = \sum_{i=1}^{N}{ \left( a_ix_i^2 + b_ix_i + c_i  + w_i f_i \right)}
\label{eq_costfunction}
\end{equation}
and the dynamical model (\ref{eq_dynamicModel}) can be rewritten as
\begin{equation}
x_i = x_i^0 + \sum\limits_{j\in\mathcal{N}_i^+}{r_{ji}f_{j}} - f_i
\label{eq_dynamic_state}
\end{equation} 

From constraints \eqref{eq_physical_constraints} and the definition of demand function, we have constraints of road capacity as follows
\begin{align}
&f_i - \sum\limits_{j\in\mathcal{N}_i^+} \underline{r}_{ji}f_j \le \underline{x}_i \label{eq_volume_down}\\
&\sum\limits_{j\in\mathcal{N}_i^+} \overline{r}_{ji}f_j - f_i \le \overline{x}_i \label{eq_volume_up}\\
&0 \le f_i \le \overline{f}_i \label{eq_flow_updown}\\
&\sum\limits_{j\in\mathcal{N}_i^+} \overline{r}_{ji}f_j \le \overline{s}_i \label{eq_flowu}
\end{align}
and the constraint for traffic light of an intersection is rewritten as
\begin{equation}
\sum_{j\in\mathcal{I}_i}v_jf_j \le T \label{eq_light}
\end{equation}
where $\underline{r}_{ji} = \underline{r}_{ji}(t)$, $\overline{r}_{ji} = \overline{r}_{ji}(t)$, $r_{ij} = r_{ij}^*(t)$, $\underline{x}_i = \rho_i(k) + \underline{\mu}_i(k)$, $\overline{x}_i = \rho_i^{cg} - \rho_i(k) - \overline{\mu}_i(k)$,$\overline{f}_i = \min\left\{d_i(\rho_i(k)), Tv_i^{-1}\right\}$, $v_i = \max_{j\in\mathcal{N}_i^-}\{r_{ij}\}\omega_i$ and $\overline{s}_i = s_i(\rho_i(k))$.
It is noted that these parameters are constant for the given cycle $t$ and known to the road cell $i$.
The equality condition \eqref{eq_dynamic_state} is for nominal case of the traffic flow, without any uncertainties.
In order to consider the uncertainties in the estimation and measurements, the lower and upper bounds of turning ratios are used in \eqref{eq_volume_down}, \eqref{eq_volume_up} and \eqref{eq_flowu} to maintain the smooth operation of traffic network.

Let $\mathcal{D}$ be the set of $\{f_i\}_{i=1,\dots,N}$ where the inequalities \eqref{eq_volume_down}, \eqref{eq_volume_up}, \eqref{eq_flow_updown}, \eqref{eq_flowu}, \eqref{eq_light} hold strictly for all $i = 1, \dots, N$.
Then $\mathcal{D}$ contains the set of the downstream traffic flows which satisfy all physical constraints.
For a feasibility of solution, we make the following assumption.
\begin{Assumption}
The set $\mathcal{D}$ is nonempty.
\label{as_feasibility}
\end{Assumption}

Formally, the objective of traffic control problem is to solve the following constrained optimization problem at the beginning of $t$-th cycle
\begin{equation}
\begin{split}
\min\limits_{x_1,\dots,x_N,f_1,\dots,f_{N}}& \Phi \textrm{ in }\eqref{eq_costfunction} \\
\textrm{s.t. }& \eqref{eq_dynamic_state}, \eqref{eq_volume_down}, \eqref{eq_volume_up}, \eqref{eq_flow_updown}, \eqref{eq_flowu}, \eqref{eq_light}
\end{split}
\label{eq_problem1}
\end{equation}

Since traffic networks are usually large scale and collected traffic data include inevitable uncertainties, traffic behaviors are very difficult to predict.
Consequently, effective traffic control strategies are required not only to be able to improve traffic behavior but also to be smartly reacting to changes in the local regions.
Compared to a centralized control system, a distributed architecture achieves better scalability and robustness.
In this setup, road cells need to communicate together to optimize overall cooperation of whole traffic network.
From equations (\ref{eq_dynamic_state}-\ref{eq_light}), one road cell $i \in \mathcal{R}$ requires its own information and its neighbors' to check feasibility constraints of decision control variable $f_i$.
Thus, it is is natural to make the following assumption on communication topology.
\begin{Assumption}
The road cell $i$ can communicate with other road cells in the sets $\mathcal{N}_i^- \cup \mathcal{N}_i^+ \cup \mathcal{I}_i$.
\label{as_communicate}
\end{Assumption}
Let $\mathcal{G} =(\mathcal{V},\mathcal{E})$ be the communication graph of the traffic network where $\mathcal{V}$ is the set of local controllers and $\mathcal{E}$ characterizes communication links in the traffic network.
Matching to the physical topology $\mathcal{P} = (\mathcal{R}, \mathcal{J})$, we have the node set $\mathcal{V} = \mathcal{R}$ and the edge set $\mathcal{E} = \left\{(i,j): i,j\in \mathcal{V}, \sigma(j) = \tau(i) \textrm{ or } \tau(j) = \sigma(i) \textrm{ or } \tau(j) = \tau(i)\right\}$ for the communication graph (see Fig.\ref{fig_4crossroad_graph}b).

The main purpose of this paper is sated as follows.
\begin{Problem}
For a given traffic network satisfying Assumption \ref{as_graph}, \ref{as_parameters}, \ref{as_feasibility} and \ref{as_communicate}, develop a control strategy which can be implemented for every road cell $i\in\mathcal{R}$ such that it uses only local information from $j\in\mathcal{N}_i^- \cup \mathcal{N}_i^+ \cup \mathcal{I}_i$ to determine its control variables $f_i$ corresponding to the optimal solution of the problem \eqref{eq_problem1}
\label{prob1}
\end{Problem}
\section{Gradient projection method}
\subsection{Primal problem and equivalent dual problem}
The Lagrangian function of the problem \eqref{eq_problem1} is given by
\begin{align*}
\mathcal{L} =& \sum_{i=1}^{N}{ \left( a_ix_i^2 + b_ix_i + c_i + w_i f_i \right) } + \sum_{i=1}^{N} \zeta_i \left( x_i^0 + \sum\limits_{j\in\mathcal{N}_i^+} r_{ji}f_j - f_i - x_i \right)\\
& + \sum_{i=1}^{N} { \lambda_i \left( f_i - \sum\limits_{j\in\mathcal{N}_i^+} \underline{r}_{ji}f_j - \underline{x}_i \right)} + \sum_{i=1}^{N} { \theta_i \left( -f_i + \sum\limits_{j\in\mathcal{N}_i^+} \overline{r}_{ji}f_j - \overline{x}_i \right)}\\
&  + \sum_{i=1}^{N} { \alpha_i \left( - f_i \right)} + \sum_{i=1}^{N} { \beta_i \left( f_i - \overline{f}_i \right)} + \sum_{i=1}^{N} { \nu_i \left( \sum\limits_{j\in\mathcal{N}_i^+} \overline{r}_{ji}f_j - \overline{s}_i\right)} + \sum_{k=1}^{N} \gamma_i \left(\sum_{j\in\mathcal{I}_i}v_jf_j - T\right)
\end{align*}
where $\zeta_i, \lambda_i, \theta_i, \alpha_i, \beta_i, \nu_i$ and $\gamma_k$ are Lagrange multipliers.
We denote $\textbf{x} = [x_1, \cdots, x_N]^T$, $\textbf{f} = [f_1, \cdots,  f_N]^T$ and stack the multipliers into vectors as $\boldsymbol{\zeta} = [\zeta_1, \cdots, \zeta_N]^T$ and $\boldsymbol{\eta} = [\boldsymbol{\eta}_1^T, \cdots, \boldsymbol{\eta}_N^T]^T$ where $\boldsymbol{\eta}_i = [\lambda_i, \theta_i, \alpha_i, \beta_i, \nu_i, \gamma_i^T$ for $i = 1,\dots,N$, then the Lagrangian function can be described in the form of $\mathcal{L} = \mathcal{L}(\textbf{x}, \textbf{f}, \boldsymbol{\zeta}, \boldsymbol{\eta})$.
It is noteworthy that the partial derivative $\nabla_{\boldsymbol{\zeta}} \mathcal{L}$ coincides with the equality constraint \eqref{eq_dynamic_state} and does not depend on multiplier $\boldsymbol{\eta}$, the partial derivative $\nabla_{\boldsymbol{\eta}} \mathcal{L}$ corresponds to the inequality constraint \eqref{eq_volume_down}-\eqref{eq_light} and does not depend on multiplier $\boldsymbol{\zeta}$.
So, we have
\begin{align*}
\mathcal{L}(\textbf{x}, \textbf{f}, \boldsymbol{\zeta}, \boldsymbol{\eta}) = \boldsymbol{\Phi}(\textbf{x},\textbf{f}) + \boldsymbol{\zeta}^T \nabla_{\boldsymbol{\zeta}} \mathcal{L}(\textbf{x}, \textbf{f}, \boldsymbol{\zeta}) + \boldsymbol{\eta}^T\nabla_{\boldsymbol{\eta}} \mathcal{L}(\textbf{x}, \textbf{f}, \boldsymbol{\eta})
\end{align*}

For simplicity, denote $\textbf{x}^0 = [x_1^0, \cdots, x_N^0]^T$, $\textbf{A} = diag(a_1,\dots,a_N)$, $\textbf{b} = [b_1,\dots,b_N]^T$, $\textbf{c} = [c_1,\dots,c_N]^T$ and $\textbf{w} = [w_1,\dots,w_N]^T$, we can write $\textbf{x} = \textbf{x}^0 - \textbf{G}^T\textbf{f}$ and $\Phi = \textbf{x}^T\textbf{A}\textbf{x} + \textbf{b}^T\textbf{x} + \textbf{c}^T + \textbf{w}^T\textbf{f}$.
Then, we have
\begin{align*}
\Phi &= \left(\textbf{x}^0 - \textbf{G}^T\textbf{f}\right)^T\textbf{A}\left(\textbf{x}^0 - \textbf{G}^T\textbf{f}\right) + \textbf{b}^T\left(\textbf{x}^0 - \textbf{G}^T\textbf{f}\right) + \textbf{c}^T + \textbf{w}^T\textbf{f}\\
&= \textbf{f}^T\textbf{G}\textbf{A}\textbf{G}^T\textbf{f} + \left(-2(\textbf{x}^0)^T\textbf{A}\textbf{G}^T - \textbf{b}^T\textbf{G}^T + \textbf{w}\right)\textbf{f} + (\textbf{x}^0)^T\textbf{A}\textbf{x}^0 + \textbf{c}^T
\end{align*}
Observe that the objective function is also a quadratic function of downstream traffic flows $\textbf{f}$.
Moreover, the constraint functions in \eqref{eq_problem1} are affine.
So, \eqref{eq_problem1} has a unique optimal solution under Assumption \ref{as_feasibility}.
Let $\textbf{x}^{opt}, \textbf{f}^{opt}, \boldsymbol{\zeta}^{opt}$ $\boldsymbol{\eta}^{opt}$ be the optimal solution of \eqref{eq_problem1}, we have the following necessary and sufficient conditions for optimality:
\begin{gather}
\nabla_{\textbf{x}} \mathcal{L}(\textbf{x}^{opt}, \textbf{f}^{opt}, \boldsymbol{\zeta}^{opt}, \boldsymbol{\eta}^{opt}) = \textbf{0}, \label{eq_KKT1}\\
\nabla_{\textbf{f}} \mathcal{L}(\textbf{x}^{opt}, \textbf{f}^{opt}, \boldsymbol{\zeta}^{opt}, \boldsymbol{\eta}^{opt}) = \textbf{0}, \label{eq_KKT2}\\
\nabla_{\boldsymbol{\zeta}} \mathcal{L}(\textbf{x}^{opt}, \textbf{f}^{opt},\boldsymbol{\zeta}^{opt}) = \textbf{0}, \label{eq_KKT3}\\
\nabla_{\boldsymbol{\eta}} \mathcal{L}(\textbf{x}^{opt}, \textbf{f}^{opt},\boldsymbol{\eta}^{opt}) \le \textbf{0}, \label{eq_KKT4}\\
\boldsymbol{\eta}^{opt} \ge \textbf{0}, \label{eq_KKT5}\\
diag\left(\boldsymbol{\eta}^{opt}\right)\nabla_{\boldsymbol{\eta}} \mathcal{L}(\textbf{x}^{opt}, \textbf{f}^{opt}, \boldsymbol{\eta}^{opt}) = \textbf{0} \label{eq_KKT6}
\end{gather}
The above equations are Karush–Kuhn–Tucker conditions of the problem \eqref{eq_problem1}.
The equations \eqref{eq_KKT1} and \eqref{eq_KKT2} are stationary conditions.
The derivative of the Lagrangian function with respect to variables $x_i$'s and $f_i$'s are required to be zero at the optimal solution.
The equations \eqref{eq_KKT3} and \eqref{eq_KKT4} are equality and inequality constraints.
For the multipliers corresponding to inequality constraints, they are necessary to be nonnegative and satisfy the complementary slackness condition \eqref{eq_KKT6}.
In \eqref{eq_KKT6}, $diag\left(\boldsymbol{\eta}^{opt}\right)$ is a diagonal matrix whose main diagonal terms consist of elements of the vector $\boldsymbol{\eta}^{opt}$.

Let $\mathcal{F}$ be a subset of $\mathcal{R}^{2N}$ where the equality constraint \eqref{eq_dynamic_state} hold for all $i = 1, \dots, N$.
\[\mathcal{F} = \{(\textbf{x}, \textbf{f}): \textbf{x} = \textbf{x}^0 - \textbf{G}^T\textbf{f}\}\]
It is no doubt that $\mathcal{F}$ is convex and the following constrained optimization is equivalent to the problem \eqref{eq_problem1}
\begin{equation}\label{eq_problem_equiprimal}
\begin{split}
\min\limits_{(\textbf{x}, \textbf{f}) \in \mathcal{F}}& \Phi(\textbf{x}, \textbf{f}) \textrm{ in }\eqref{eq_costfunction} \\
\textrm{s.t. }& \eqref{eq_volume_down}, \eqref{eq_volume_up}, \eqref{eq_flow_updown}, \eqref{eq_flowu}, \eqref{eq_light}
\end{split}
\end{equation} 
We have the Lagrangian function of the problem \eqref{eq_problem_equiprimal} as
\begin{align*}
\mathcal{L}^e(\textbf{x}, \textbf{f}, \boldsymbol{\eta}) = \Phi(\textbf{x},\textbf{f}) + \boldsymbol{\eta}^T\nabla_{\boldsymbol{\eta}} \mathcal{L}(\textbf{x}, \textbf{f}, \boldsymbol{\eta})
\end{align*}
and define the dual function $\Psi(\boldsymbol{\eta}) = \inf_{(\textbf{x}, \textbf{f}) \in \mathcal{F}}\mathcal{L}^e(\textbf{x}, \textbf{f}, \boldsymbol{\eta})$.
Then, the dual problem corresponging to the constrained optimization problem \eqref{eq_problem_equiprimal} is
\[\max\limits_{\boldsymbol{\eta} \ge \textbf{0}}\Psi(\boldsymbol{\eta})\]
Under Assumption \ref{as_feasibility}, \eqref{eq_problem_equiprimal} satisfies Assumption \ref{as_ConvexAndInterior}.
According to Proposition \ref{prob_strongdual} given in the Appendix-A, there exists at least one Lagrange multiplier $\boldsymbol{\eta}^{*}$ such that
\[\boldsymbol{\eta}^{*} = \arg\max_{\boldsymbol{\eta}\ge\textbf{0}}\Psi(\boldsymbol{\eta}) \ge \textbf{0},\textrm{ and } \Phi^{*} = \inf_{(\textbf{x}, \textbf{f}) \in \mathcal{F}}\mathcal{L}^e(\textbf{x}, \textbf{f},  \boldsymbol{\eta}^{*})\]
where $\Phi^{*}$ is the optimal value of \eqref{eq_problem_equiprimal}.
Moreover, if 
\begin{equation}\label{eq_dualproperty}
diag\left(\boldsymbol{\eta}^{*}\right)\nabla_{\boldsymbol{\eta}} \mathcal{L}^e (\textbf{x}^{*}, \textbf{f}^{*}, \boldsymbol{\eta}^{*}) = \textbf{0}
\end{equation}
where $(\textbf{x}^{*},\textbf{f}^{*}) = \arg\min_{(\textbf{x}, \textbf{f}) \in \mathcal{F}}{\mathcal{L}^e (\textbf{x}, \textbf{f}, \boldsymbol{\eta}^{*})}$, then $(\textbf{x}^{*},\textbf{f}^{*})$ is the optimal solution of the primal problem \eqref{eq_problem_equiprimal} as stated in Proposition \ref{prob_optimalsolution} (see Appendix-A).
That means $(\textbf{x}^{*}, \textbf{f}^{*}) = (\textbf{x}^{opt}, \textbf{f}^{opt})$ because the problem \eqref{eq_problem1} has a unique optimal solution and so is the problem \eqref{eq_problem_equiprimal}.
Since $\nabla_{\boldsymbol{\zeta}} \mathcal{L}(\textbf{x}, \textbf{f}, \boldsymbol{\zeta}) = \textbf{0}$ for all $(\textbf{x}, \textbf{f}) \in \mathcal{F}$, we have $\inf_{(\textbf{x}, \textbf{f}) \in \mathcal{F}}\mathcal{L}^e(\textbf{x}, \textbf{f}, \boldsymbol{\eta}) \equiv \inf_{(\textbf{x}, \textbf{f}) \in \mathcal{F}, \boldsymbol{\zeta}}\mathcal{L}(\textbf{x}, \textbf{f}, \boldsymbol{\zeta}, \boldsymbol{\eta})$ for all $\boldsymbol{\eta}$.
It implies
\begin{equation}\label{eq_dualproblem1}
\boldsymbol{\eta}^{*} = \arg\max_{\boldsymbol{\eta}\ge\textbf{0}}\inf_{(\textbf{x}, \textbf{f}) \in \mathcal{F}, \boldsymbol{\zeta}}\mathcal{L}(\textbf{x}, \textbf{f}, \boldsymbol{\zeta}, \boldsymbol{\eta})
\end{equation}
and there is a vector $\boldsymbol{\zeta}^*$ such that
\begin{equation}\label{eq_dualproblem2}
(\textbf{x}^{*},\textbf{f}^{*}, \boldsymbol{\zeta}^*) = \arg\min_{(\textbf{x}, \textbf{f}) \in \mathcal{F}, \boldsymbol{\zeta}}{\mathcal{L} (\textbf{x}, \textbf{f}, \boldsymbol{\zeta}, \boldsymbol{\eta}^{*})}
\end{equation}
\subsection{Finding Lagrangian mulitplier based on the gradient projection method}
Fixing $\boldsymbol{\eta}$ and taking the derivative of the Lagrangian function with respect to the traffic volumes, the traffic downstream flows and Lagrangian multipliers corresponding to equality constraints, we have
\begin{subequations}
\begin{align}
\frac{\partial \mathcal{L}}{\partial x_i} =& 2a_ix_i + b_i - \zeta_i\\
\frac{\partial \mathcal{L}}{\partial f_i} =&  - \zeta_i + \sum\limits_{j\in\mathcal{N}_i^-}r_{ij}\zeta_j + h_i(\boldsymbol{\eta})\\
\frac{\partial \mathcal{L}}{\partial \zeta_i} =& x_i^0 + \sum\limits_{j\in\mathcal{N}_i^+} r_{ji}f_j - f_i - x_i
\end{align}
\label{eq_Lagrang_derivative1}
\end{subequations}
where $h_i(\boldsymbol{\eta}) = w_i + \lambda_i - \sum\limits_{j\in\mathcal{N}_i^-} \underline{r}_{ij}\lambda_j - \theta_i + \sum\limits_{j\in\mathcal{N}_i^-} \overline{r}_{ij}\theta_j - \alpha_i + \beta_i + \sum\limits_{j\in\mathcal{N}_i^-}\overline{r}_{ij}\nu_j + \sum\limits_{j\in\mathcal{I}_i}v_j\gamma_j$.

Let $\textbf{x}^*(\boldsymbol{\eta})$, $\textbf{f}^*(\boldsymbol{\eta})$ and $\boldsymbol{\zeta}^*(\boldsymbol{\eta})$ be the optimal solution of $\inf_{(\textbf{x}, \textbf{f}) \in \mathcal{F}}\mathcal{L}(\textbf{x}, \textbf{f}, \boldsymbol{\zeta}, \boldsymbol{\eta})$ for a given $\boldsymbol{\eta}$.
From the optimality conditions which are obtained by equalizing the right-hand sides of \eqref{eq_Lagrang_derivative1} to zero, we have
\begin{align}
\textbf{G}\boldsymbol{\zeta}^*(\boldsymbol{\eta}) &= \textbf{h}(\boldsymbol{\eta}) \label{eq_relationship2}\\
x_i^*(\zeta_i^*) &= \frac{\zeta_i^*-b_i}{2a_i} \label{eq_relationship1}\\
\textbf{G}^T\textbf{f}^*(\boldsymbol{\eta}) &= \textbf{x}^0 - \textbf{x}^*(\boldsymbol{\eta}) \label{eq_relationship3}
\end{align}
where $\textbf{h}(\boldsymbol{\eta}) = [h_1(\boldsymbol{\eta}), \cdots, h_N(\boldsymbol{\eta})]^T \in \mathbb{R}^N$. 
From the definition of $h_i$, we can rewrite $\textbf{h} = \textbf{H}\boldsymbol{\eta} + \textbf{w}$ with a proper matrix $\textbf{H}$.
Since $\textbf{G}$ is a nonsingular matrix, the optimal traffic flow vector can be described in a linear form:
\begin{equation}
\textbf{f}^*(\boldsymbol{\eta}) = \textbf{G}^{-T}(\textbf{x}^0 - \textbf{x}^*(\boldsymbol{\eta})) = \textbf{P}\boldsymbol{\eta} + \textbf{p}
\end{equation}
where $\textbf{P} = -0.5\textbf{G}^{-T} \textbf{A}^{-1} \textbf{G}^{-1}\textbf{H}$ and $\textbf{p} = \textbf{G}^{-T}\textbf{x}^0 + 0.5\textbf{G}^{-T}\textbf{A}^{-1}\textbf{b} - \textbf{G}^{-T}\textbf{w}$.
It is easy to check that all the element matrices in the matrix $\textbf{P}$ have finite norm; so we have its norm $\delta = ||\textbf{P}|| < \infty$.
Thus, it holds
\[||\textbf{f}^*(\boldsymbol{\eta}^{(1)}) - \textbf{f}^*(\boldsymbol{\eta}^{(2)})|| = ||\textbf{P}(\boldsymbol{\eta}^{(1)} - \boldsymbol{\eta}^{(2)})|| \le \delta ||\boldsymbol{\eta}^{(1)} - \boldsymbol{\eta}^{(2)}||.\]

Since $\textbf{x}^*(\boldsymbol{\eta})$, $\textbf{f}^*(\boldsymbol{\eta})$, $\boldsymbol{\zeta}^*(\boldsymbol{\eta})$ are continuous function with respect to $\boldsymbol{\eta}$, we have $\Psi(\boldsymbol{\eta}) = \mathcal{L}(\textbf{x}^*(\boldsymbol{\eta}), \textbf{f}^*(\boldsymbol{\eta}), \boldsymbol{\zeta}^*(\boldsymbol{\eta}), \boldsymbol{\eta})$
is also continuous and 
\begin{align*}
\nabla_{\boldsymbol{\eta}}\Psi(\boldsymbol{\eta}) =& \frac{\partial \mathcal{L}(\textbf{x}^*(\boldsymbol{\eta}))}{\partial \textbf{x}} \frac{\partial \textbf{x}^*(\boldsymbol{\eta})}{\partial \boldsymbol{\eta}} + \frac{\partial \mathcal{L}(\textbf{f}^*(\boldsymbol{\eta}))}{\partial \textbf{f}} \frac{\partial \textbf{f}^*(\boldsymbol{\eta})}{\partial \boldsymbol{\eta}} + \frac{\partial \mathcal{L}(\boldsymbol{\zeta}^*(\boldsymbol{\eta}))}{\partial \boldsymbol{\zeta}(\boldsymbol{\eta})} \frac{\partial \boldsymbol{\zeta}^*(\boldsymbol{\eta})}{\partial \boldsymbol{\eta}} + \nabla_{\boldsymbol{\eta}} \mathcal{L}(\boldsymbol{\eta})
\end{align*}
From \eqref{eq_relationship2}, \eqref{eq_relationship1} and \eqref{eq_relationship3}, we have $\frac{\partial \mathcal{L}}{\partial \textbf{x}} (\textbf{x}^*(\boldsymbol{\eta})) = \frac{\partial \mathcal{L}}{\partial \textbf{f}} (\textbf{f}^*(\boldsymbol{\eta})) = \frac{\partial \mathcal{L}}{\partial \boldsymbol{\zeta}} (\boldsymbol{\zeta}^*(\boldsymbol{\eta})) = 0$.
This implies $\nabla \Psi(\boldsymbol{\eta}) = \nabla_{\boldsymbol{\eta}} \mathcal{L}(\boldsymbol{\eta})$.
The partial derivatives of the Lagrangian function with respect to multipliers corresponding to the following inequality constraints are given as follows:
\begin{subequations}\label{eq_Lagrang_derivative2}
\begin{align}
\frac{\partial \mathcal{L}}{\partial \gamma_i} =& \sum_{j\in\mathcal{I}_i}v_jf_j^* - T\\
\frac{\partial \mathcal{L}}{\partial \nu_i} =& \sum_{j\in\mathcal{N}_i^+}{\overline{r}_{ji}f_j^*} - \overline{s}_i\\
\frac{\partial \mathcal{L}}{\partial \lambda_i} =& f_i^* - \sum\limits_{j\in\mathcal{N}_i^+} \underline{r}_{ji}f_j^* - \underline{x}_i\\
\frac{\partial \mathcal{L}}{\partial \theta_i} =& -f_i^* + \sum\limits_{j\in\mathcal{N}_i^+} \overline{r}_{ji}f_j^* - \overline{x}_i\\
\frac{\partial \mathcal{L}}{\partial \alpha_i} =& - f_i^*\\
\frac{\partial \mathcal{L}}{\partial \beta_i} =& f_i^* - \overline{f}_i
\end{align}
\end{subequations}
So, $\nabla \Psi(\boldsymbol{\eta})$ also has the linear form of $\nabla \Psi(\boldsymbol{\eta}) = \textbf{Q}\textbf{f}^*(\boldsymbol{\eta}) + \textbf{q}$ where $||\textbf{Q}|| = \varrho < \infty$.
The following inequality holds
\begin{align*}
||\nabla \Psi(\boldsymbol{\eta}^{(1)}) - \nabla \Psi(\boldsymbol{\eta}^{(2)})|| &\le \varrho ||\textbf{f}^*(\boldsymbol{\eta}^{(1)}) - \textbf{f}^*(\boldsymbol{\eta}^{(2)})|| \le \varrho\delta ||\boldsymbol{\eta}^{(1)} - \boldsymbol{\eta}^{(2)}||
\end{align*}
Moreover, the function $\Psi(\boldsymbol{\eta})$ is concave or $-\Psi(\boldsymbol{\eta})$ is convex.
According to Proposition \ref{prob_gradientproject},
under the gradient projection update law \eqref{eq_gradientprojection}
\begin{equation}
\boldsymbol{\eta}(k+1) = \max\left\{\textbf{0}, \boldsymbol{\eta}(k) + \epsilon \nabla \Psi(\boldsymbol{\eta}(k)) \right\} \label{eq_gradientprojection}
\end{equation}
with a step size $\epsilon < \frac{2}{\varrho\delta}$, we have $\lim_{k \rightarrow \infty}\boldsymbol{\eta}(k) = \boldsymbol{\eta}^*$ where 
\begin{equation}\label{eq_confirm}
\begin{split}
\boldsymbol{\eta}^* &= \max\left\{\textbf{0}, \boldsymbol{\eta}^* + \epsilon \nabla \Psi(\boldsymbol{\eta}^*) \right\}\\
& = \max\left\{\textbf{0}, \boldsymbol{\eta}^* + \epsilon \nabla_{\boldsymbol{\eta}} \mathcal{L}(\textbf{x}(\boldsymbol{\eta}^*), \textbf{f}(\boldsymbol{\eta}^*), \boldsymbol{\eta}^*) \right\}
\end{split}
\end{equation}
From the above analysis, we propose the update rule \eqref{eq_centralized_flow}-\eqref{eq_centralized_multipliers} to estimate the optimal solution $\textbf{f}^{opt}$.
\begin{align}
\textbf{f}(k+1) &= \textbf{P}\boldsymbol{\eta}(k) + \textbf{p} \label{eq_centralized_flow}\\
\boldsymbol{\eta}(k+1) &= \max\left\{\textbf{0}, \boldsymbol{\eta}(k) + \epsilon \left(\textbf{Q}\textbf{f}(k+1) + \textbf{q}\right)\right\} \label{eq_centralized_multipliers}
\end{align}

\begin{Theorem}\label{th_centralized}
An optimal solution of the constrained minimization problem \eqref{eq_problem1} is asymptotically achieved by applying the update rule \eqref{eq_centralized_flow}-\eqref{eq_centralized_multipliers} with a sufficiently small $\epsilon$ in the sense:
\[\textbf{f}^{opt} = \lim_{k \rightarrow \infty}\textbf{f}(k).\]
\end{Theorem}
\begin{proof}
In \eqref{eq_centralized_flow}, $\textbf{f}(k+1) = \textbf{f}^*(\boldsymbol{\eta}(k))$ is the optimal solution of the subproblem $\min_{(\textbf{x},\textbf{f})\in\,\mathcal{F}, \boldsymbol{\zeta}}\mathcal{L}(\textbf{x},\textbf{f},\boldsymbol{\zeta},\boldsymbol{\eta}(k))$.
So, \eqref{eq_centralized_multipliers} is coincided with the gradient projection update law \eqref{eq_gradientprojection}.
As $k$ goes to infinite, $\boldsymbol{\eta}(k)$ converges to $\boldsymbol{\eta}^{*}$.

Let $\eta_i^*$ be the $i$-th element of the vector $\boldsymbol{\eta}^*$ and let $\nabla_i$ be the $i$-th element of the vector $\nabla_{\boldsymbol{\eta}} \mathcal{L}(\textbf{x}(\boldsymbol{\eta}^*), \textbf{f}(\boldsymbol{\eta}^*), \boldsymbol{\eta}^*)$.
From \eqref{eq_confirm}, we have $\eta_i^* \ge 0$ and 
\begin{equation}\label{eq_result}
\nabla_i = \max\{0, \eta_i^*+\epsilon\nabla_i\}, \epsilon>0
\end{equation}
for all $i$.
If $\eta_i^* > 0$, the equation \eqref{eq_result} holds if and only if $\nabla_i = 0$.
If $\eta_i^* = 0$, the equation \eqref{eq_result} holds if and only if $\nabla_i \le 0$.
These facts imply that $\eta_i^*\nabla_i = 0$ and $\nabla_i \le 0$ for all $i$.
So, $(\textbf{x}^{*}(\boldsymbol{\eta}^{*}), \textbf{f}^{*}(\boldsymbol{\eta}^{*}), \boldsymbol{\zeta}^{*}(\boldsymbol{\eta}^{*}), \boldsymbol{\eta}^{*})$ satisfies \eqref{eq_KKT4}, \eqref{eq_KKT5},\eqref{eq_KKT6}.

In addition, the equations \eqref{eq_relationship1}, \eqref{eq_relationship2}, \eqref{eq_relationship3} guarantee the condition \eqref{eq_KKT1}, \eqref{eq_KKT2},\eqref{eq_KKT3}, respectively. 
Thus, all KKT conditions are satisfied or $(\textbf{x}^{*}(\boldsymbol{\eta}^{*}), \textbf{f}^{*}(\boldsymbol{\eta}^{*}), \boldsymbol{\zeta}^{*}(\boldsymbol{\eta}^{*}), \boldsymbol{\eta}^{*})$ is the optimal solution of \eqref{eq_problem1}. 
\end{proof}

For more detail, we provide Algorithm \ref{alg_Centralized} as the process of the centralized method to find the optimal solution of \eqref{eq_problem1}.
The stopping criteria can be chosen as 
\[||\boldsymbol{\eta}(k+1)-\boldsymbol{\eta}(k)|| < \Delta\]
for a sufficiently small $\Delta$.
\begin{algorithm}[htb]
\begin{algorithmic}[1]
\BState \textbf{Input}: Parameters of road cells $i \in \mathcal{R}$: $a_i, b_i, w_i, d_i$ and $r_{ij}, \underline{r}_{ij}, \overline{r}_{ij}$, $j\in\mathcal{N}_i^+$.
\BState \textbf{Output}: Optimal solution $\textbf{f}^{opt}$.
\BState \emph{Initialization}:  
\State $\quad$ Construct the matrices $\textbf{P}, \textbf{Q}$ and vectors $\textbf{p}, \textbf{q}$.
\State $\quad$ $\boldsymbol{\eta} \gets \textbf{0}$ and choose $\epsilon < \frac{2}{||\textbf{P}||.||\textbf{Q}||}$.
\BState \emph{Iterative update}:
\State $\quad$\textbf{while} the stopping criteria is not satisfied \textbf{do}
\State $\qquad$ $\textbf{f}^* \gets \textbf{P}\boldsymbol{\eta} + \textbf{p}$
\State $\qquad$ $\boldsymbol{\eta} \gets \max\left\{\textbf{0}, \boldsymbol{\eta} + \epsilon \left(\textbf{Q}\textbf{f}^* + \textbf{q}\right)\right\}$
\State $\quad$\textbf{end while}
\BState \emph{Finish algorithm}: $\textbf{f}^{opt} \gets \textbf{f}^*$
\end{algorithmic}
\caption{Centralized algorithm to solve Problem \ref{prob1}.}
\label{alg_Centralized}
\end{algorithm}
\section{Distributed algorithm for traffic control}
Before describing a distributed version of Algorithm \ref{alg_Centralized}, we first design a distributed algorithm which can be applied for a finite-time solving \eqref{eq_relationship2} and \eqref{eq_relationship3} to reduce the running time of the while loops.
\subsection{Minimum-time final value computation}
The analysis in this subsection is similar to the minimal time consensus method presented in \cite{ShreyasSundaram2007, YeYuan2009, YeYuan2013}.
However, different from these works, we are not restrict to a consensus problem where discrete state matrix is a stochastic matrix. 
Instead, we consider the discrete-time model
\begin{subequations}\label{eq_dis_dynamics}
\begin{align}
\textbf{x}(l+1) &= \textbf{M} \textbf{x}(l) + \textbf{m}, t \ge 1 \\
y_r(l) &= \textbf{e}_n^T[r] \textbf{x}(l).
\end{align}
\end{subequations}
where $\textbf{M} \in \mathbb{R}^{n \times n}$ is strictly stable, $\textbf{m} \in \mathbb{R}^n$ is a constant vector and $\textbf{e}_n[r]$ is an $n$-dimensional vector with all elements being zero except the $r$-th element being one.
In \eqref{eq_dis_dynamics}, $\textbf{x} \in \mathbb{R}^n$ is a variable vector and $y_r$ is an observer corresponding to the $r$-th element in $\textbf{x}$. 
Let $\textbf{z}(l) = \textbf{x}(l+1) - \textbf{x}(l)$ be the difference between two consecutive iterations.
From \eqref{eq_dis_dynamics}, it is easy to verify that
\[\textbf{z} (l+1) = \textbf{M}\textbf{z} (l)\]
Notice that $z_r(l+i) = \textbf{e}_n^T[r] \textbf{M}^i\textbf{z}(k) = y_r(l+i+1)-y_r(l+i)$.
Since $\textbf{M}$ has only stable eigenvalues, we have $\lim_{l \rightarrow \infty} \textbf{z} (l) = \textbf{0}_n$.
It implies the convergence for $y_r(l)$, i.e., $y_r^{\infty} = \lim_{l\rightarrow\infty}y_r(l)$, exists.

For a matrix pair $[\textbf{M},\textbf{e}_n^T[r]]$, there exists the minimal polynomial $q_r(t) = t^{\varsigma_r+1} + \sum_{i=0}^{\varsigma_r}{\Theta_i t^i}$ with minimum degree $\varsigma_r + 1 \le n$ that satisfies $\textbf{e}_n^T[r] q_r(\textbf{M}) = \textbf{0}_n^T$.
So, we have
\[\textbf{e}_n^T[r] q_r(\textbf{M})\textbf{z}(l) = 0 = \textbf{e}_n^T[r]\sum_{i=0}^{\varsigma_r+1}{\Theta_i \textbf{M}^i}\textbf{z}(l), \textrm{ } \Theta_{\varsigma_r+1} = 1.\]
The above equation is equivalent to 
\begin{align*}
0 = & \Theta_0 [y_r(l+1)-y_r(l)] + \Theta_1 [y_r(l+2)-y_r(l+1)] + \cdots + \Theta_{d_r}[y_r(l+\varsigma_r+1)-y_r(l+\varsigma_r)] + \\
& + [y_r(l+\varsigma_r+2)-y_r(l+\varsigma_r+1)]
\end{align*}
or ${y_r(l)}\Theta_0 + {y_r(l+1)}\Theta_1 + \cdots + {y_r(l+\varsigma_r+1)}\Theta_{\varsigma_r+1} = const$ for all $l \ge 0$.
Thus, the final value of observer \eqref{eq_dis_dynamics}b is computed as
\begin{equation}
y_r^{\infty} = \frac{{y_r(1)}\Theta_0 + {y_r(2)}\Theta_1 + \cdots + {y_r(\varsigma_r+2)}\Theta_{\varsigma_r+1}}{\Theta_0 + \Theta_1 + \cdots + \Theta_{\varsigma_r+1}}. \label{eq_final_value}
\end{equation}
Let $q(t)$ be the polynomial of the matrix $\textbf{M}$.
We have $q(t) = \prod_{i = 1}^{n}(t-\kappa_i)$ where $\kappa_i$ is an eigenvalue of $\textbf{M}$.
Then one is not a root of $q(t)$.
In \cite{ShreyasSundaram2007}, the minimal polynomial $q_r(t)$ divides the minimal polynomial of $\textbf{M}$.
That means all roots of $q_r(t)$ are also roots of $q(t)$.
It guarantees that all roots of $q_r(t)$ are different from one or the denominator of \eqref{eq_final_value} is nonzero, $q_r(1) = \Theta_0 + \Theta_1 + \cdots + \Theta_{\varsigma_r+1} \neq 0$.

\begin{algorithm}[htb]
\begin{algorithmic}[1]
\BState \textbf{Input}: Successive observations of $y_r(i), l = 0,1,\dots$.
\BState \textbf{Output}: Final value $y_r^\infty$.
\BState \emph{Initialization}: $l=1$
\[\textbf{Y} \gets \frac{1}{{z}_r(0)}, \textbf{c} \gets {z}_r(1), s \gets {z}_r(2)-\frac{({z}_r(1))^2}{{z}_r(0)}.\] 
\BState \emph{First defective Hankel detection loop}:
\State $\quad$\textbf{while} $s \neq 0$ \textbf{do}
\State $\qquad$ $l \gets l+1$
\State $\qquad$ $\textbf{Y} \gets \left[ {\begin{array}{*{20}{c}} \textbf{Y}  + \textbf{Y}\textbf{c}s^{-1}\textbf{c}^T \textbf{Y} & - \textbf{Y}\textbf{c} s^{-1} \\ -s^{-1}\textbf{c}^T \textbf{Y} & s^{-1} \end{array}} \right]$
\State $\qquad$ $\textbf{c} \gets \left[ {\begin{array}{*{20}{c}} {z}_r(l) & \cdots & {z}_r(2l-1) \end{array}} \right]^T$
\State $\qquad$ $s \gets {z}_r(2l)-\textbf{c}^T\textbf{Y}\textbf{c}$
\State $\quad$\textbf{end while}
\BState \emph{Coefficients computation} $\Theta \gets \left[ {\begin{array}{*{20}{c}} -\textbf{Y}\textbf{c}\\ 1 \end{array}} \right]$.
\BState \emph{Final value computation} 
\State $\quad$ $k \gets 0, D_r\gets l-2$, $\gamma \gets \Theta$, $y_r^\infty \gets$ \eqref{eq_final_value}
\end{algorithmic}
\caption{Minimum-time computation of the final value.}
\label{alg_min_cal}
\end{algorithm}

In \cite{YeYuan2013}, Yuan et al. developed an algorithm to calculate the coefficients $\Theta$'s using only observations $y_r$'s.
The key idea is to find the first defective Hankel matrix given in \eqref{eq_Hankel}.
\begin{equation}
\textbf{H}_{{z}_r}[l] = \left[ {\begin{array}{*{20}{c}} {z}_r(0) & {z}_r(1) & \cdots & {z}_r(l)\\ {z}_r(1) & {z}_r(2) & \dots & {z}_r(l+1)\\ \vdots & \vdots & \ddots & \vdots\\ {z}_r(l) & {z}_r(l+1) & \dots & {z}_r(2l) \end{array}} \right]
\label{eq_Hankel}
\end{equation} 
We here provide Algorithm \ref{alg_min_cal} to determine the first index $l$ where $\textbf{H}_{{z}_r}[l]$ looses its rank and the corresponding normalized kernel $\boldsymbol{\Theta}$.
That means
\begin{align}
rank (\textbf{H}_{{z}_r}[l])& = rank (\textbf{H}_{{z}_r}[l+h]), \textrm{ } h = 1,2,\dots\\
&\textbf{H}_{{z}_r}[l]\boldsymbol{\Theta} = \textbf{0}.
\end{align}
Denote $\textbf{c}_l  = \left[ {\begin{array}{*{20}{c}} {{z}_r(l)}&{{z}_r(l+1)}& \cdots &{{z}_r(2l-1)} \end{array}} \right]^T$ and $\textbf{Z}_l = \textbf{H}_{{z}_r}[l]$, we have
\begin{equation}
\textbf{Z}_l = \left[ {\begin{array}{*{20}{c}} \textbf{Z}_{l-1} & \textbf{c}_{l} \\ \textbf{c}_{l}^T & {z}_r(2l) \end{array}} \right] \textrm{ for } l>1
\label{eq_Hankel_block}
\end{equation} 
Consider the Schur complement $s_l = {z}_r(2l) - \textbf{c}_l^T \textbf{Z}_{l-1}^{-1} \textbf{c}_l$, it is easy to verify that if $s_k \neq 0$ and $\textbf{Z}_{k-1}$ is full rank, then the matrix $\textbf{Z}_k$ is nonsingular and its inverse matrix is given as
\begin{equation}
\textbf{Z}_{l}^{-1} = \left[ {\begin{array}{*{20}{c}} \textbf{Z}_{l-1}^{-1} (\textbf{I}_l  + \textbf{c}_{l}s_{l}^{-1}\textbf{c}_{l}^T \textbf{Z}_{l-1}^{-1}) & -\textbf{Z}_{l-1}^{-1}\textbf{c}_{l} s_l^{-1} \\ -s_l^{-1}\textbf{c}_{l}^T\textbf{Z}_{l-1}^{-1} & s_l^{-1} \end{array}} \right]
\label{eq_invert_Hankel}
\end{equation}
It implies that $\textbf{Z}_{l}$ will not loose the rank until $s_l=0$.
Or, the minimum index $l^*$ where $s_{l^*}=0$ determines the first defective Hankel matrix.
From \eqref{eq_Hankel_block}, we have 
\begin{equation}
\left[ {\begin{array}{*{20}{c}} \Theta_0& \Theta_1 & \cdots & \Theta_{l^*-1} \end{array}} \right]^T = -\textbf{Z}_{l^*-1}^{-1} \textbf{c}_{l^*}
\end{equation}
sine $\textbf{Z}_{l^*}\boldsymbol{\Theta} = \textbf{0}$ and $\Theta_{l^*} = 1$.
\begin{Remark}
Since Algorithm \ref{alg_min_cal} requires only observations $y_r(l)$, it can run in parallel with the updated process \eqref{eq_dis_dynamics}.
\end{Remark}
\subsection{Main algorithm}
Recalling that $\textbf{G} = \textbf{I}-\textbf{R}$ where $\textbf{R}$ is a strictly stable matrix, the equations \eqref{eq_relationship2} and \eqref{eq_relationship3} can be solved by applying the Jacobi method as follows.
\begin{align*}
\boldsymbol{\zeta}(l+1) &= \textbf{R}\boldsymbol{\zeta}(l) + \textbf{h}(\boldsymbol{\eta}(k))\\
\textbf{f}(l+1) &= \textbf{R}^T\textbf{f}(l) + \textbf{x}^0 - \textbf{x}^*(\boldsymbol{\eta}(k))
\end{align*}
We have $\lim_{l\rightarrow\infty} \boldsymbol{\zeta}(l) = \boldsymbol{\zeta}^*(\boldsymbol{\eta}(k))$ and $\lim_{k\rightarrow\infty} \textbf{f}(l) = \textbf{f}^*(\boldsymbol{\eta}(k))$.
Another advantage of the Jacobi method is allowing the road cells to choose initial estimation states arbitrarily.
The detailed formulations of this application for each road cell $i \in \mathcal{R}$ are given as
\begin{align}
\zeta_i(l+1) &= \sum\limits_{j\in\mathcal{N}_i^-}r_{ij}\zeta_j(l) + h_i(\boldsymbol{\eta}(k)) \label{eq_zetaUpdate}\\
f_i(l+1) &= \sum\limits_{j\in\mathcal{N}_i^+}r_{ji}f_j(l) + x_i^0 - x_i^*(\boldsymbol{\eta}(k)) \label{eq_flowUpdate}
\end{align}
Now from the determined multipliers corresponding to equality constraints $\zeta_i^*(\eta(k))$'s and downstream traffic flows $f_i^*(\eta(k))$'s, the multipliers corresponding to the inequality constraints $\boldsymbol{\eta}_i$'s are updated by \eqref{eq_eta_update}.
\begin{subequations}
\begin{align}
\lambda_i(k+1) &= \max\left\{0, \lambda_i(k) + \epsilon \left(f_i^*(\boldsymbol{\eta}(k)) - \sum\limits_{j\in\mathcal{N}_i^+} \underline{r}_{ji}f_j^*(\boldsymbol{\eta}(k)) - \underline{x}_i\right) \right\}\\
\theta_i(k+1) &= \max\left\{0, \theta_i(k) + \epsilon \left(-f_i^*(\boldsymbol{\eta}(k)) + \sum\limits_{j\in\mathcal{N}_i^+} \overline{r}_{ji}f_j^*(\boldsymbol{\eta}(k)) - \overline{x}_i\right) \right\}\\
\alpha_i(k+1) &= \max\left\{0, \alpha_i(k) - \epsilon f_i^*(\boldsymbol{\eta}(k)) \right\}\\
\beta_i(k+1) &= \max\left\{0, \beta_i(k) + \epsilon \left(f_i^*(\boldsymbol{\eta}(k)) - \overline{f}_i\right) \right\}\\
\nu_i(k+1) &= \max\left\{0, \nu_i(k) + \epsilon \left(\sum_{i \in \mathcal{N}_i^+} \overline{r}_{ji}f_j^*(\boldsymbol{\eta}(k)) - \overline{s}_i\right) \right\}\\
\gamma_i(k+1) &= \max\left\{0, \gamma_i(k) + \epsilon \left(\sum_{j \in \mathcal{I}_i}v_jf_j^*(\boldsymbol{\eta}(k)) - T\right) \right\}
\end{align}
\label{eq_eta_update}
\end{subequations}
Note that update rules \eqref{eq_zetaUpdate}, \eqref{eq_flowUpdate}, \eqref{eq_eta_update} requires information of road cell $i$ and its neighboring road cells in $\mathcal{N}_i^+ \cup \mathcal{N}_i^- \cup \mathcal{I}_i$.

Moreover, by applying Algorithm \ref{alg_min_cal}, to each road cell $i$ we are able to estimate coefficient vectors  $\boldsymbol{\Theta}^{(\zeta_i)} = [\Theta_0^{(\zeta_i)}, \Theta_1^{(\zeta_i)}, \cdots, \Theta_{\varsigma_i}^{(\zeta_i)}, 1 ]^T$ and $\boldsymbol{\Theta}^{(f_i)} = [\Theta_0^{(f_i)}, \Theta_1^{(f_i)}, \cdots, \Theta_{\varsigma_i'}^{(f_i)}, 1 ]^T$ such that
\begin{align}
\zeta_i^*(\boldsymbol{\eta}(t)) &= \frac{\sum\limits_{j=0}^{\varsigma_i+1}\Theta_j^{(\zeta_i)}\zeta_i(l+j)}{\sum\limits_{j=0}^{\varsigma_i+1}\Theta_j^{(\zeta_i)}} \label{eq_opt_zeta}\\
f_i^*(\boldsymbol{\eta}(t)) &= \frac{\sum\limits_{j=0}^{\varsigma_i'+1}\Theta_j^{(f_i)}f_i(l+j)}{\sum\limits_{j=0}^{\varsigma_i'+1}\Theta_j^{(f_i)}} \label{eq_opt_f}
\end{align} 
with the data of $\zeta_i(l)$'s obtained in \eqref{eq_zetaUpdate} (resp., $f_i(l)$'s obtained in \eqref{eq_flowUpdate}).
In \cite{ThemistoklisCharalambous2015}, Charalambous et al. show that there exists a set of initial states ($\zeta_i(0)$'s, $f_i(0)$'s) of measure zero for which $\boldsymbol{\Theta}$ is different from $\boldsymbol{\Theta}^{(\zeta_i)}$ (resp., $\boldsymbol{\Theta}^{(f_i)}$).
The main reason for which the algorithm fails is because the Hankel matrix \eqref{eq_Hankel} loses rank for the first time too early.
Although it is hard to characterize such the set of initial states, practical techniques to alleviate the problem are available; see, e.g. Remark 7 and Remark 8 in \cite{ThemistoklisCharalambous2015}.
We here suggest a strategy to deal with this issue.
That is, once a road cell has finished Algorithm \ref{alg_min_cal}, it uses the determined vector $\boldsymbol{\Theta}$ to compute the final state as in \eqref{eq_final_value} and sends this value to its neighbors with a special flag.
By this way, every road cell knows the final values belonging to itself and its neighbors'.
Then, it can check whether its local equation (\eqref{eq_zetaUpdate} or \eqref{eq_flowUpdate}) is satisfied.
If there is any violation determined by a road cell $i$, it sends a special message to notify its neighbors to run Algorithm \ref{alg_min_cal} again with other initial states.

\begin{algorithm}[htb]
\begin{algorithmic}[1]
\BState \emph{Initialize:} Apply Algorithm \ref{alg_min_cal}
\State $\quad$ Determine $\boldsymbol{\Theta}^{(\zeta_i)}, \boldsymbol{\Theta}^{(f_i)}$ and number $D_i, D_i'$.
\State $\quad$ $D_{max} = \max\left\{\left\{D_i\right\}_{i = 1,\dots,N}, \left\{D_i'\right\}_{i = 1,\dots,N}\right\}$
\BState \emph{Iterative update}:
\State $\quad$\textbf{Initialization} $k = 0, \boldsymbol{\eta}_i(0) \gets \textbf{0}$
\State $\quad$\textbf{while} the stop criteria is not satisfied \textbf{do}
\State $\qquad$ Run \eqref{eq_zetaUpdate} and \eqref{eq_flowUpdate} in $D_{max}$ times, then compute
\State $\qquad$ $\qquad$ $\zeta_i^*(\boldsymbol{\eta}(k)) \gets$ \eqref{eq_opt_zeta},
\State $\qquad$ $\qquad$ $x_i^*(\boldsymbol{\eta}(k)) \gets$ \eqref{eq_relationship1},
\State $\qquad$ $\qquad$ $f_i^*(\boldsymbol{\eta}(k)) \gets$ \eqref{eq_opt_f}.
\State $\qquad$ $\boldsymbol{\eta}_i(k+1) \gets $\eqref{eq_eta_update}
\State $\qquad$ $k \gets k+1$
\State $\quad$\textbf{end while}
\BState \emph{Finish algorithm } $f_i^{opt} \gets f_i(k)$
\end{algorithmic}
\caption{Proposed control strategy running in the beginning of cycle.}
\label{alg_proposed_control}
\end{algorithm}
Finally, we summarize our analysis in this section as Algorithm \ref{alg_proposed_control} which is our main result.
Since this algorithm is a distributed version of Algorithm \ref{alg_Centralized}, its correctness is guaranteed by Theorem \ref{th_centralized}.
\begin{Theorem}\label{th_distributed}
By applying Algorithm \ref{alg_proposed_control}, each road cell $i$ is able to determine its downstream traffic flow corresponding to $i$-th element in the optimal solution of problem \eqref{eq_problem1} by using only local information.
\end{Theorem}

It is noteworthy that Algorithm \ref{alg_proposed_control} has key properties of distributed strategies.
At the beginning of each cycle, the data of current traffic conditions are used to determine downstream traffic flows.
The computational load of each road cell $i \in \mathcal{R}$ does not depend on the size of the network and it is required to communicate with some special road cells, which have same sink node or source node as $i$.
Moreover, if there is any structural change at one road cell (added or removed), only this road cell and its neighbors need to be reprogrammed. 
In the case one road cell $j$ cannot be used due to accident or other reasons, Algorithm \ref{alg_proposed_control} is still applicable for remaining road cells with the setup of the turning split ratio $r_{ij}(t) = 0$ for all $i \in \mathcal{R}$.
\color{black}
\section{Numerical Simulation}
In this section, some numerical simulations are conducted to validate our proposed algorithms by using MATLAB.

\begin{figure*}[htb]
\begin{center}
\includegraphics[width=0.42\textwidth]{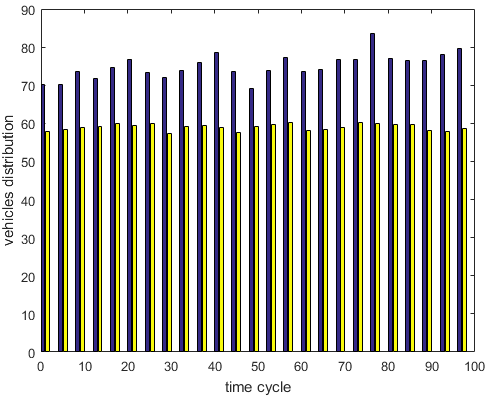}
\includegraphics[width=0.43\textwidth]{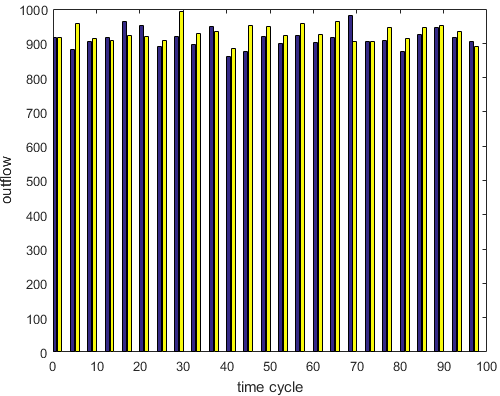}
\caption{Simulation results for normal inflow level case. Yellow bars are cost measurements applied Algorithm \ref{alg_proposed_control} and mazarine bars are corresponding to fixed time strategy.}
\label{fig_normallevel}
\end{center}
\end{figure*}
\begin{figure*}[htb]
\begin{center}
\includegraphics[width=0.42\textwidth]{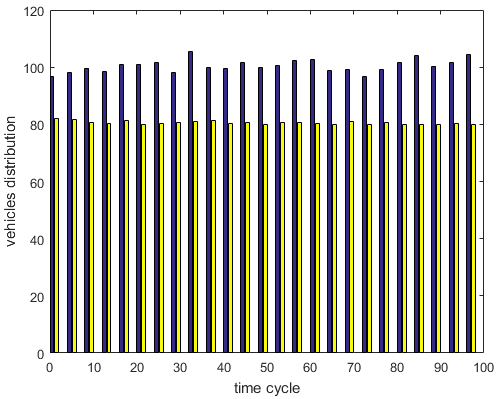}
\includegraphics[width=0.42\textwidth]{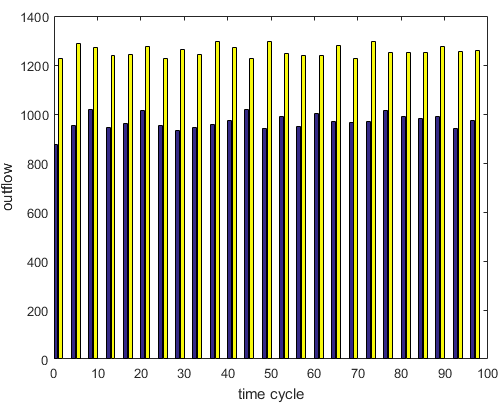}
\caption{Simulation results for high input level case. Yellow bars are cost measurements applied Algorithm \ref{alg_proposed_control} and mazarine bars are corresponding to fixed time strategy.} 
\label{fig_highlevel}
\end{center}
\end{figure*}
To evaluate the effectiveness of Algorithm \ref{alg_proposed_control}, we apply it for the traffic network consisting of $2\times 2$ intersections shown in Fig. \ref{fig_4crossroad_graph}.
The obtained results are compared with traffic behaviors under fixed time control strategy, where green duration times of intersections are constant in all cycles.
The objective function of road cell $i \in \mathcal{R}$ is considered as
\[\Phi_i(\rho_i(k), \psi_i(k)) = a_i\sum_{i\in\mathcal{R}}a_i(\rho_i(k))^2 + w_i\psi_i(k)\] 
where $a_i = 0.55$ if $\sigma(i) = O$ and $a_i = 0.5$ otherwise; $w_i = -20$ if $\tau(i) = O$ and $w_i = -10$ otherwise (negative weight for the purpose of maximizing).
It is a weighted sum of the vehicles volume squared and the downstream traffic flow for each road cell.
Assume that traffic congestion volumes $\rho_i^{cg} = 300$ for all $i$ and the exogenous inflow entering into road cell $i$, where $\sigma(i) = O$, is 
\[\mu_i(t) = Q_{in}[1+0.1rand()]\]
where $Q_{in}$ corresponds to inflow level.

The traffic behaviors we choose to compare are the averaged cost for vehicle distributions, which is $\sum_{i\in\mathcal{R}}a_i(\rho_i(k))^2/ \sum_{i\in\mathcal{R}}\rho_i(k)$, and the sum of the downstream traffic flow of destination road cells $i$'s where $\tau(i) = O$, i.e., the total outflows or the vehicles throughput of traffic network.
In the case of normal level $Q_{in} = \rho_i^{cg}/3$, simulation results are given in Fig. \ref{fig_normallevel}.
The left figure represents the averaged cost for vehicle distributions of some time cycles from $1$ to $100$, and the right one corresponds to the total outflows.
We see that although the vehicles throughput of traffic network in two strategies are almost similar, Algorithm \ref{alg_proposed_control} is shown to improve vehicle distribution condition.
This reduces the risk of congestion and increase the smooth operation of overall traffic network.
Fig. \ref{fig_highlevel} are simulation result for high level case where $Q_{in} = \rho_i^{cg}/2$. 
By applying Algorithm \ref{alg_proposed_control}, traffic conditions are shown to be improved significantly in both criteria for vehicle distribution and total vehicle throughput of the traffic network.

\begin{figure}[htb]
\begin{center}
\includegraphics[width=0.35\textwidth]{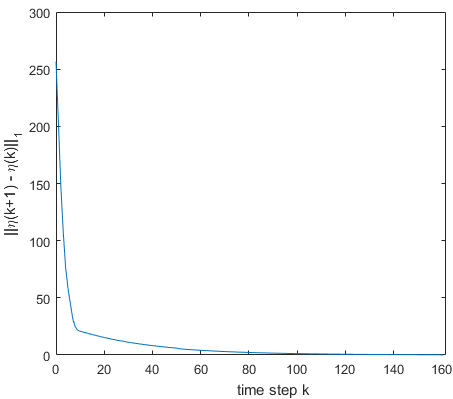}
\caption{The evolution of Lagrangian multiplier $\boldsymbol{\eta}$.} 
\label{fig_eta}
\end{center}
\end{figure}
Consider the operation of while loop in Algorithm \ref{alg_proposed_control}, the main problem is to update iteratively multiplier $\boldsymbol{\eta}$ of inequality constraints in \eqref{eq_problem1}.
We illustrate the evolution of $\boldsymbol{\eta}$ in $60$-th cycle as in Fig.\ref{fig_eta}.
The difference of $\boldsymbol{\eta}(k+1)-\boldsymbol{\eta}(k)$ is very tiny after about $100$ steps or $\boldsymbol{\eta}^{opt}$ and the optimal solution of \eqref{eq_problem1} is estimated with sufficient accuracy by running about $100$ while loops.

We determine the maximum number $D_{max}$ of updates \eqref{eq_flowUpdate} (or \eqref{eq_zetaUpdate}) required by one road cell applying Algorithm \ref{alg_min_cal} for the traffic network of $m \times n$ intersections (shown in Fig. \ref{fig_mncrossroad_graph}).
To check the advantage of using Algorithm \ref{alg_min_cal}, we compare these upper bounds with the necessary number of update \eqref{eq_flowUpdate}(or \eqref{eq_zetaUpdate}) to have sufficient closed solution of \eqref{eq_relationship2} (or \eqref{eq_relationship3}, resp.).
TABLE \ref{tbl} shows comparison result for some $m$'s and $n$'s.
By applying Algorithm \ref{alg_min_cal}, the running time of each while loop (in Algorithm \ref{alg_proposed_control}) is reduced significantly as the size of traffic networks increases.
\begin{figure}[htb]
\begin{center}
\centering
\scalebox{0.7}{\begin{tikzpicture}[
textnode/.style={rectangle},
squarenode/.style={rectangle, draw=black, fill=white,  very thick, minimum size=10mm},
ellipsenodeIn1/.style={ellipse, draw=black, fill=black, minimum height=4mm, minimum width= 2mm},
ellipsenodeIn2/.style={ellipse, draw=black, fill=black, minimum width=4mm, minimum height= 2mm},
ellipsenodeOut1/.style={ellipse, draw=black, fill=white, minimum height=4mm, minimum width= 2mm},
ellipsenodeOut2/.style={ellipse, draw=black, fill=white, minimum width=4mm, minimum height= 2mm},
roundnode/.style={circle, draw=black, fill=black!20,  very thick, minimum size=3mm},
]
\node at (1,5) [squarenode] (n11) {$I_{1,1}$};
\node at (3,5) [squarenode] (n12) {$I_{1,2}$};
\node at (5,5) [squarenode] (n1n1) {$I_{1,n-1}$};
\node at (7,5) [squarenode] (n1n) {$I_{1,n}$};
\node at (1,3) [squarenode] (n21) {$I_{2,1}$};
\node at (3,3) [squarenode] (n22) {$I_{2,2}$};
\node at (5,3) [squarenode] (n2n1) {$I_{2,n-1}$};
\node at (7,3) [squarenode] (n2n) {$I_{2,n}$};
\node at (1,1) [squarenode] (nm1) {$I_{m,1}$};
\node at (3,1) [squarenode] (nm2) {$I_{m,2}$};
\node at (5,1) [squarenode] (nmn1) {$I_{m,n-1}$};
\node at (7,1) [squarenode] (nmn) {$I_{m,n}$};
\node at (0.7,6.5) [ellipsenodeIn2] (no1) {};
\node at (1.3,6.5) [ellipsenodeOut2] (no2) {};
\node at (2.7,6.5) [ellipsenodeIn2] (no3) {};
\node at (3.3,6.5) [ellipsenodeOut2] (no4) {};
\node at (4.7,6.5) [ellipsenodeIn2] (no5) {};
\node at (5.3,6.5) [ellipsenodeOut2] (no6) {};
\node at (6.7,6.5) [ellipsenodeIn2] (no7) {};
\node at (7.3,6.5) [ellipsenodeOut2] (no8) {};
\node at (-0.5,5.3) [ellipsenodeOut1] (no9) {};
\node at (8.5,5.3) [ellipsenodeIn1] (no10) {};
\node at (-0.5,4.7) [ellipsenodeIn1] (no11) {};
\node at (8.5,4.7) [ellipsenodeOut1] (no12) {};
\node at (-0.5,3.3) [ellipsenodeOut1] (no13) {};
\node at (8.5,3.3) [ellipsenodeIn1] (no14) {};
\node at (-0.5,2.7) [ellipsenodeIn1] (no15) {};
\node at (8.5,2.7) [ellipsenodeOut1] (no16) {};
\node at (-0.5,1.3) [ellipsenodeOut1] (no17) {};
\node at (8.5,1.3) [ellipsenodeIn1] (no18) {};
\node at (-0.5,0.7) [ellipsenodeIn1] (no19) {};
\node at (8.5,0.7) [ellipsenodeOut1] (no20) {};
\node at (0.7,-0.5) [ellipsenodeOut2] (no21) {};
\node at (1.3,-0.5) [ellipsenodeIn2] (no22) {};
\node at (2.7,-0.5) [ellipsenodeOut2] (no23) {};
\node at (3.3,-0.5) [ellipsenodeIn2] (no24) {};
\node at (4.7,-0.5) [ellipsenodeOut2] (no25) {};
\node at (5.3,-0.5) [ellipsenodeIn2] (no26) {};
\node at (6.7,-0.5) [ellipsenodeOut2] (no27) {};
\node at (7.3,-0.5) [ellipsenodeIn2] (no28) {};
\draw[->,{line width=3pt},black!40] (no1.south)--(no1.south|-n11.north);
\draw[<-,{line width=3pt},black!40] (no2.south)--(no2.south|-n11.north);
\draw[->,{line width=3pt},black!40] (no3.south)--(no3.south|-n12.north);
\draw[<-,{line width=3pt},black!40] (no4.south)--(no4.south|-n12.north);
\draw[->,{line width=3pt},black!40] (no5.south)--(no5.south|-n1n1.north);
\draw[<-,{line width=3pt},black!40] (no6.south)--(no6.south|-n1n1.north);
\draw[->,{line width=3pt},black!40] (no7.south)--(no7.south|-n1n.north);
\draw[<-,{line width=3pt},black!40] (no8.south)--(no8.south|-n1n.north);

\draw[<-,{line width=3pt},black!40] (no9.east)--(no9.east-|n11.west);
\draw[<-,{line width=3pt},black!40, transform canvas={yshift=3mm}] (n11.east)->(n12.west);
\draw[<-,{line width=3pt},black!40, dotted, transform canvas={yshift=3mm}] (n12.east)->(n1n1.west);
\draw[<-,{line width=3pt},black!40, transform canvas={yshift=3mm}] (n1n1.east)->(n1n.west);
\draw[<-,{line width=3pt},black!40, transform canvas={yshift=3mm}] (n1n.east)->(n1n.east-|no10.west);

\draw[->,{line width=3pt},black!40] (no11.east)--(no11.east-|n11.west);
\draw[->,{line width=3pt},black!40, transform canvas={yshift=-3mm}] (n11.east)->(n12.west);
\draw[->,{line width=3pt},black!40, dotted, transform canvas={yshift=-3mm}] (n12.east)->(n1n1.west);
\draw[->,{line width=3pt},black!40, transform canvas={yshift=-3mm}] (n1n1.east)->(n1n.west);
\draw[->,{line width=3pt},black!40, transform canvas={yshift=-3mm}] (n1n.east)--(n1n.east-|no12.west);

\draw[->,{line width=3pt},black!40, transform canvas={xshift=-3mm}] (n11.south)--(n11.south|-n21.north);
\draw[<-,{line width=3pt},black!40, transform canvas={xshift=3mm}] (n11.south)--(n11.south|-n21.north);
\draw[->,{line width=3pt},black!40, transform canvas={xshift=-3mm}] (n12.south)--(n12.south|-n22.north);
\draw[<-,{line width=3pt},black!40, transform canvas={xshift=3mm}] (n12.south)--(n12.south|-n22.north);
\draw[->,{line width=3pt},black!40, transform canvas={xshift=-3mm}] (n1n1.south)--(n1n1.south|-n2n1.north);
\draw[<-,{line width=3pt},black!40, transform canvas={xshift=3mm}] (n1n1.south)--(n1n1.south|-n2n1.north);
\draw[->,{line width=3pt},black!40, transform canvas={xshift=-3mm}] (n1n.south)--(n1n.south|-n2n.north);
\draw[<-,{line width=3pt},black!40, transform canvas={xshift=3mm}] (n1n.south)--(n1n.south|-n2n.north);

\draw[<-,{line width=3pt},black!40] (no13.east)--(no13.east-|n11.west);
\draw[<-,{line width=3pt},black!40, transform canvas={yshift=3mm}] (n21.east)->(n22.west);
\draw[<-,{line width=3pt},black!40, dotted, transform canvas={yshift=3mm}] (n22.east)->(n2n1.west);
\draw[<-,{line width=3pt},black!40, transform canvas={yshift=3mm}] (n2n1.east)->(n2n.west);
\draw[<-,{line width=3pt},black!40, transform canvas={yshift=3mm}] (n2n.east)--(n2n.east-|no14.west);

\draw[->,{line width=3pt},black!40] (no15.east)--(no15.east-|n21.west);
\draw[->,{line width=3pt},black!40, transform canvas={yshift=-3mm}] (n21.east)->(n22.west);
\draw[->,{line width=3pt},black!40, dotted, transform canvas={yshift=-3mm}] (n22.east)->(n2n1.west);
\draw[->,{line width=3pt},black!40, transform canvas={yshift=-3mm}] (n2n1.east)->(n2n.west);
\draw[->,{line width=3pt},black!40, transform canvas={yshift=-3mm}] (n2n.east)--(n2n.east-|no16.west);

\draw[->,{line width=3pt},black!40,dotted, transform canvas={xshift=-3mm}] (n21.south)--(n21.south|-nm1.north);
\draw[<-,{line width=3pt},black!40,dotted, transform canvas={xshift=3mm}] (n21.south)--(n21.south|-nm1.north);
\draw[->,{line width=3pt},black!40,dotted, transform canvas={xshift=-3mm}] (n22.south)--(n22.south|-nm2.north);
\draw[<-,{line width=3pt},black!40,dotted, transform canvas={xshift=3mm}] (n22.south)--(n22.south|-nm2.north);
\draw[->,{line width=3pt},black!40,dotted, transform canvas={xshift=-3mm}] (n2n1.south)--(n2n1.south|-nmn1.north);
\draw[<-,{line width=3pt},black!40,dotted, transform canvas={xshift=3mm}] (n2n1.south)--(n2n1.south|-nmn1.north);
\draw[->,{line width=3pt},black!40,dotted, transform canvas={xshift=-3mm}] (n2n.south)--(n2n.south|-nmn.north);
\draw[<-,{line width=3pt},black!40,dotted, transform canvas={xshift=3mm}] (n2n.south)--(n2n.south|-nmn.north);

\draw[<-,{line width=3pt},black!40] (no17.east)--(no17.east-|nm1.west);
\draw[<-,{line width=3pt},black!40, transform canvas={yshift=3mm}] (nm1.east)->(nm2.west);
\draw[<-,{line width=3pt},black!40, dotted, transform canvas={yshift=3mm}] (nm2.east)->(nmn1.west);
\draw[<-,{line width=3pt},black!40, transform canvas={yshift=3mm}] (nmn1.east)->(nmn.west);
\draw[<-,{line width=3pt},black!40, transform canvas={yshift=3mm}] (nmn.east)--(nmn.east-|no18.west);

\draw[->,{line width=3pt},black!40] (no19.east)--(no19.east-|nm1.west);
\draw[->,{line width=3pt},black!40, transform canvas={yshift=-3mm}] (nm1.east)->(nm2.west);
\draw[->,{line width=3pt},black!40, dotted, transform canvas={yshift=-3mm}] (nm2.east)->(nmn1.west);
\draw[->,{line width=3pt},black!40, transform canvas={yshift=-3mm}] (nmn1.east)->(nmn.west);
\draw[->,{line width=3pt},black!40, transform canvas={yshift=-3mm}] (nmn.east)--(nmn.east-|no20.west);

\draw[->,{line width=3pt},black!40, transform canvas={xshift=-3mm}] (nm1.south)--(nm1.south|-no21.north);
\draw[<-,{line width=3pt},black!40, transform canvas={xshift=3mm}] (nm1.south)--(nm1.south|-no22.north);
\draw[->,{line width=3pt},black!40, transform canvas={xshift=-3mm}] (nm2.south)--(nm2.south|-no23.north);
\draw[<-,{line width=3pt},black!40, transform canvas={xshift=3mm}] (nm2.south)--(nm2.south|-no24.north);
\draw[->,{line width=3pt},black!40, transform canvas={xshift=-3mm}] (nmn1.south)--(nmn1.south|-no25.north);
\draw[<-,{line width=3pt},black!40, transform canvas={xshift=3mm}] (nmn1.south)--(nmn1.south|-no26.north);
\draw[->,{line width=3pt},black!40, transform canvas={xshift=-3mm}] (nmn.south)--(nmn.south|-no27.north);
\draw[<-,{line width=3pt},black!40, transform canvas={xshift=3mm}] (nmn.south)--(nmn.south|-no28.north);
\end{tikzpicture}}
\end{center}
\caption{Graph representation for traffic network of $m\times n$ intersections.}
\label{fig_mncrossroad_graph}
\end{figure}
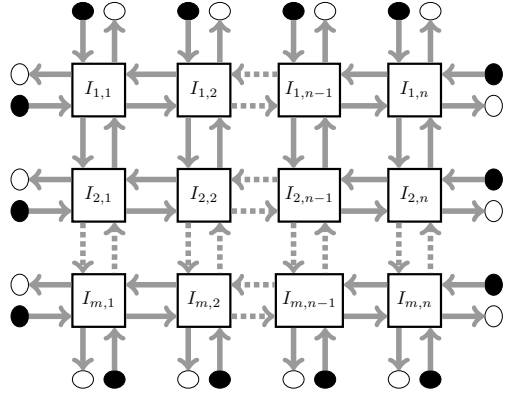
\begin{table}[htb]
\centering
\begin{tabular}{c|c|c|c|c|c|c|c}
\multicolumn{8}{c}{$m = 2$}\\
 \hline
$n = $  & $2$ & $5$ & $10$ & $20$ & $30$ & $45$ & $60$\\
 \hline
 \eqref{eq_zetaUpdate} or \eqref{eq_flowUpdate} & $15$ & $26$ & $35$ & $40$ & $42$ & $42$ & $42$ \\
 \hline
 Algorithm \ref{alg_min_cal} & $14$ & $24$ & $28$ & $30$ & $30$ & $32$ & $32$ \\
 \hline
\multicolumn{8}{c}{$m = 5$}\\
 \hline
$n = $  & $2$ & $5$ & $10$ & $20$ & $30$ & $45$ & $60$\\
 \hline
 \eqref{eq_zetaUpdate} or \eqref{eq_flowUpdate} & $26$ & $58$ & $84$ & $90$ & $98$ & $112$ & $128$ \\
 \hline
 Algorithm \ref{alg_min_cal} & $24$ & $30$ & $34$ & $34$ & $36$ & $38$ & $40$ \\
 \hline
\multicolumn{8}{c}{$m = 10$}\\
 \hline
$n = $  & $2$ & $5$ & $10$ & $20$ & $30$ & $45$ & $60$\\
 \hline
 \eqref{eq_zetaUpdate} or \eqref{eq_flowUpdate} & $35$ & $84$ & $125$ & $145$ & $152$ & $159$ & $157$ \\
 \hline
 Algorithm \ref{alg_min_cal} & $28$ & $34$ & $44$ & $46$ & $48$ & $52$ & $56$ \\
 \hline
\multicolumn{8}{c}{$m = 20$}\\
 \hline
$n = $  & $2$ & $5$ & $10$ & $20$ & $30$ & $45$ & $60$\\
 \hline
 \eqref{eq_zetaUpdate} or \eqref{eq_flowUpdate} & $40$ & $90$ & $145$ & $192$ & $268$ & $368$ & $456$ \\
 \hline
 Algorithm \ref{alg_min_cal} & $30$ & $34$ & $46$ & $56$ & $78$ & $98$ & $118$ \\
 \hline
\end{tabular}
\caption{Comparison of observation times used with and without Algorithm \ref{alg_min_cal}.}
\label{tbl}
\end{table}
\section{Conclusion}
The traffic management is always a crucial problem in human life and affects (directly or indirectly) to economics, society, politics and environment.
To cope with large-scale traffic network, this paper proposed a distributed control strategy to coordinate the traffic flows.
We used discrete-time Cell Transmission Model (CTM) to describe the dynamic nature of road cells and formulated the control objective as a constrained minimization problem in a multiagent perspective.
Under our proposed algorithms, every road cell uses only local information to determine its control decision (downstream traffic flow) to improve the traffic conditions in the overall network. 
We showed that the obtained downstream traffic flows are feasible in the sense that they satisfy all physical constraints.
\bibliographystyle{IEEEtran}
\bibliography{mylib}

\begin{thebibliography}{10}
\providecommand{\url}[1]{#1}
\csname url@samestyle\endcsname
\providecommand{\newblock}{\relax}
\providecommand{\bibinfo}[2]{#2}
\providecommand{\BIBentrySTDinterwordspacing}{\spaceskip=0pt\relax}
\providecommand{\BIBentryALTinterwordstretchfactor}{4}
\providecommand{\BIBentryALTinterwordspacing}{\spaceskip=\fontdimen2\font plus
\BIBentryALTinterwordstretchfactor\fontdimen3\font minus
  \fontdimen4\font\relax}
\providecommand{\BIBforeignlanguage}[2]{{%
\expandafter\ifx\csname l@#1\endcsname\relax
\typeout{** WARNING: IEEEtran.bst: No hyphenation pattern has been}%
\typeout{** loaded for the language `#1'. Using the pattern for}%
\typeout{** the default language instead.}%
\else
\language=\csname l@#1\endcsname
\fi
#2}}
\providecommand{\BIBdecl}{\relax}
\BIBdecl

\bibitem{GlenWeisbrod2001}
G.~Treyz, G.~E. Weisbrod, and D.~Vary, \emph{Economic Implications of
  Congestion}.\hskip 1em plus 0.5em minus 0.4em\relax National Cooperative
  Highway Research Program and American Association of State Highway and
  Transportation Officials and National Research Council (U.S.). Transportation
  Research Board, 2001.

\bibitem{KaiZhang2013}
K.~Zhanga and S.~Batterman, ``Air pollution and health risks due to vehicle
  traffic,'' \emph{Science of The Total Environment}, vol. 450-451, p.
  307–316, 2013.

\bibitem{MarkosPapageorgiou2003}
M.~Papageorgiou, C.~Diakaki, V.~Dinopoulou, A.~Kotsialos, and Y.~Wang, ``Review
  of road traffic control strategies,'' \emph{Proceedings of the IEEE},
  vol.~91, no.~12, pp. 4416--4426, Dec 2003.

\bibitem{LeiChen2016}
L.~Chen and C.~Englund, ``Cooperative intersection management: A survey,''
  \emph{IEEE Transactions on Intelligent Transportation Systems}, vol.~17,
  no.~2, pp. 570 -- 586, Feb 2016.

\bibitem{KonstantinosAmpountolas2009}
K.~Aboudolas, M.~Papageorgiou, and E.~Kosmatopoulos, ``Store-and-forward based
  methods for the signal control problem in large-scale congested urban road
  networks,'' \emph{Transportation Research Part C: Emerging Technologies},
  vol.~18, no.~5, p. 680–694, 2010.

\bibitem{KonstantinosAmpountolas2010}
------, ``A rolling-horizon quadratic programming approach to the signal
  control problem in large scale congested urban road networks,''
  \emph{Transportation Research Part C: Emerging Technologies}, vol.~17, no.~2,
  p. 163–174, 2009.

\bibitem{JackHaddad2010}
J.~Haddad, B.~D. Schutter, D.~Mahalel, I.~Ioslovich, and P.~Gutman, ``Optimal
  steady-state control for isolated traffic intersections,'' \emph{IEEE
  Transactions on Automatic Control}, vol.~55, no.~11, pp. 2612 -- 2617, Nov
  2010.

\bibitem{SamuelCoogan2017}
S.~Coogan, E.~Kim, G.~Gomes, M.~Arcak, and P.~Varaiya, ``Offset optimization in
  signalized traffic networks via semidefinite relaxation,''
  \emph{Transportation Research Part B: Methodological}, vol. 100, p. 82–92,
  2017.

\bibitem{EricSKim2017}
E.~S. Kim, C.-J. Wu, R.~Horowitz, and M.~Arcak, ``Offset optimization of
  signalized intersections via the burer-monteiro method,'' in \emph{Proc. of
  the 2017 American Control Conference (ACC)}, 2017, p. 3554–3559.

\bibitem{SteliosTimotheou2015}
S.~Timotheou, C.~G. Panayiotou, and M.~M. Polycarpou, ``Distributed traffic
  signal control using the cell transmission model via the alternating
  direction method of multipliers,'' \emph{IEEE Transactions on Intelligent
  Transportation Systems}, vol.~16, no.~2, pp. 919 -- 933, Apr 2015.

\bibitem{PietroGrandinetti2018}
P.~Grandinetti, C.~C. de~Wit, and F.~Garin, ``Distributed optimal traffic
  lights design for large-scale urban networks,'' \emph{IEEE Transactions on
  Control Systems Technology}, pp. 1 -- 14, 2018.

\bibitem{CarlosFDaganzo1994}
C.~F. Daganzo, ``The cell transmission model: A dynamic representation of
  highway traffic consistent with the hydrodynamic theory,''
  \emph{Transportation Research Part B: Methodological}, vol.~28, no.~4, p.
  269–287, 1994.

\bibitem{CarlosFDaganzo1995}
------, ``The cell transmission model, part ii: network traffic,''
  \emph{Transportation Research Part B: Methodological}, vol.~29, no.~2, p.
  79–93, 1995.

\bibitem{VietHoangPham2019}
V.~H. Pham, K.~Sakurama, and H.-S. Ahn, ``A decentralized control strategy for
  urban traffic network,'' in \emph{Proc. of the 12th Asian Control Conference
  (ASCC)}, 2019.

\bibitem{ShreyasSundaram2007}
S.~Sundaram and C.~N. Hadjicostis, ``Finite-time distributed consensus in
  graphs with time-invariant topologies,'' in \emph{Proc. of the 2007 American
  Control Conference (ACC)}, 2007.

\bibitem{YeYuan2009}
Y.~Yuan, G.-B. Stan, L.~Shi, and J.~Goncalves1, ``Decentralised final value
  theorem for discrete-time lti systems with application to minimal-time
  distributed consensus,'' in \emph{Proc. of the 48th IEEE Conference on
  Decision and Control and 28th Chinese Control Conference}, 2009.

\bibitem{YeYuan2013}
Y.~Yuan, G.-B. Stan, L.~Shi, M.~Barahona, and J.~Goncalves, ``Decentralised
  minimum-time consensus,'' \emph{Automatica}, vol.~49, no.~11, pp. 1227--1235,
  2013.

\bibitem{ChangJenLan1999}
C.~J. Lan and G.~A. Davis, ``Real-time estimation of turning movement
  proportions frompartial counts on urban networks,'' \emph{Transportation
  Research Part C}, vol.~7, pp. 305 -- 327, 1999.

\bibitem{MartinRodriguezVega2019}
M.~Rodriguez-Vega, C.~C. de~Wit, and H.~Fourati, ``Location of turning ratio
  and flow sensors for flow reconstruction in large traffic networks,''
  \emph{Transportation Research Part B}, vol. 121, p. 21–40, 2019.

\bibitem{AshishBhaskar2015}
A.~Bhaskar, M.~Qu, and E.~Chung, ``Bluetooth vehicle trajectory by fusing
  bluetooth and loops: Motorway travel time statistics,'' \emph{Transportation
  Research Part B}, vol. 121, p. 21–40, 2019.

\bibitem{AfshinAbadi2015}
A.~Abadi, T.~Rajabioun, and P.~A. Ioannou, ``Traffic flow prediction for road
  transportation networks with limited traffic data,'' \emph{IEEE Transactions
  on Intelligent Transportation Systems}, vol.~16, no.~2, pp. 653 -- 662, April
  2015.

\bibitem{SeunghyeonLee2015}
S.~Lee, S.~C. Wong, C.~C.~C. Pang, and K.~Choi, ``Real-time estimation of
  lane-to-lane turning flows at isolated signalized junctions,'' \emph{IEEE
  Transactions on Intelligent Transportation Systems}, vol.~16, no.~3, pp. 1549
  -- 1558, June 2015.

\bibitem{ThemistoklisCharalambous2015}
T.~Charalambous, Y.~Yuan, T.~Yang, W.~Pan, C.~N. Hadjicostis, and M.~Johansson,
  ``Distributed finite-time average consensus in digraphs in the presence of
  time delays,'' \emph{IEEE Transaction on Control of Network Systems}, vol.~2,
  no.~4, pp. 370--381, Dec 2015.

\bibitem{DimitriPBertsekas1999}
D.~P. Bertsekas, \emph{Nonlinear Programming}, 2nd~ed.\hskip 1em plus 0.5em
  minus 0.4em\relax Academic Press, 1999.

\end{thebibliography}
\appendix
In this part, we provide some mathematical background for our analysis.
\subsection{Convex optimization}
Consider the nonlinear constrained problem
\begin{equation}\label{eq_ap_problem1}
\begin{split}
\min\limits_{\textbf{x}\in\mathcal{X}}&\textrm{ } \phi(\textbf{x})\\
\textrm{s.t. }& g_i(\textbf{x}) \le 0, i = 1,\dots,p
\end{split}
\end{equation}
under convexity and interior point assumptions as follows.
\begin{Assumption}[Assumption 5.3.2 \cite{DimitriPBertsekas1999}]\label{as_ConvexAndInterior}
The set $\mathcal{X}$ is a convex subset of $\mathbb{R}^n$ and the function $\phi:\mathbb{R}^n \rightarrow \mathbb{R}$, $g_j:\mathbb{R}^n \rightarrow \mathbb{R}$ are convex over $\mathcal{X}$.
In addition, there exists a vector $\bar{\textbf{x}}\in\mathcal{X}$ such that $g_i(\bar{\textbf{x}}) < 0, i = 1,\dots,p$.
\end{Assumption}
We have the Lagrangian function as
\[\mathcal{L}(\textbf{x},\boldsymbol{\mu}) = \phi(\textbf{x}) + \boldsymbol{\mu}^T\textbf{g}(\textbf{x})\]
Define the dual function by $\varphi(\boldsymbol{\mu}) = \inf_{\textbf{x}\in\mathcal{X}}\mathcal{L}(\textbf{x},\boldsymbol{\mu})$, then the associated dual problem of \eqref{eq_ap_problem1} is given as follows.
\begin{equation}\label{eq_ap_dualproblem}
\max\limits_{\boldsymbol{\mu} \ge \textbf{0}}\varphi(\boldsymbol{\mu})
\end{equation}
Let $\phi^*$ and $\varphi^*$ be the optimal value of the primal problem \eqref{eq_ap_problem1} and the dual problem \eqref{eq_ap_dualproblem}, respectively.
From the definition, we have $\varphi^* \le \phi^*$ in general.
Moreover, if the primal problem \eqref{eq_ap_problem1} satisfies conditions stated in Assumption \ref{as_ConvexAndInterior}, the existence of Lagrange multiplier is guaranteed.
\begin{Proposition}[Proposition 5.3.2 \cite{DimitriPBertsekas1999}]\label{prob_strongdual}
Let Assumption \ref{as_ConvexAndInterior} hold for the problem \eqref{eq_ap_problem1}.
Then $\phi^* = \varphi^*$ and there exists at least one Lagrange multiplier $\boldsymbol{\mu}^*$ where 
\[\boldsymbol{\mu}^* \ge \textbf{0} \textrm{ and } \phi^* = \inf\limits_{\textbf{x}\in\mathcal{X}}\mathcal{L}(\textbf{x},\boldsymbol{\mu}^*)\]
\end{Proposition}
When the Lagrange multiplier $\boldsymbol{\mu}^*$ is determined, the optimal solution of \eqref{eq_ap_problem1} is verified by the following proposition.
\begin{Proposition}[Proposition 5.1.11 \cite{DimitriPBertsekas1999}]\label{prob_optimalsolution}
Let $\boldsymbol{\mu}^*$ be a Lagrange multiplier having the properties given in Proposition \ref{prob_strongdual}. 
Then $\textbf{x}^*$ is a global minimum of the primal problem if and only if $\textbf{x}^*$ is feasible and 
\[\textbf{x}^* = \arg\min\limits_{\textbf{x}\in\mathcal{X}}\mathcal{L}(\textbf{x}, \boldsymbol{\mu}^*)\]
\[\mu_j^*g_j(\textbf{x}^*) = 0, j = 1,\dots,p\]
\end{Proposition}
\subsection{Gradient projection methods}
Consider the constrained optimization problem 
\begin{equation}\label{eq_ap_problem2}
\begin{split}
\textrm{min }& \phi(\textbf{x})\\
\textrm{s.t. }& \textbf{x} \in \mathcal{X}
\end{split}
\end{equation}
where $\mathcal{X}$ is a nonempty and convex subset of $\mathcal{R}^n$ and $\phi:\mathcal{R}^n \rightarrow \mathcal{R}$ is continuously differentiable over $\mathcal{X}$.
The gradient projection method is a feasible direction method of the form 
\[\textbf{x}^{k+1} = \textbf{x}^k + \alpha^k(\bar{\textbf{x}}^k - \textbf{x}^k)\]
where $\bar{\textbf{x}}^k = [\textbf{x}^k - \epsilon^k\nabla\phi(\textbf{x}^k)]^+$.
Here, $[\cdot]^+$ denotes projection on the set $\mathcal{X}$, $\alpha^k \in (0,1]$ is a stepsize, and $\epsilon^k$ is a positive scalar. 
The simplest choice for the stepsize $\alpha^k$ and the scalar $\epsilon^k$ are $\alpha^k = 1$ and $\epsilon^k = \epsilon = const$.
Under this setup, the convergence of estimation $\textbf{x}^k$ to the stationary point $\tilde{\textbf{x}}$ where $[\tilde{\textbf{x}} - \epsilon \nabla \phi(\tilde{\textbf{x}})]^+ = \tilde{\textbf{x}}$ is sated as in Proposition \ref{prob_gradientproject}.
\begin{Proposition}[Proposition 2.3.2 \cite{DimitriPBertsekas1999}]\label{prob_gradientproject}
Let $\textbf{x}^k$ be a sequence generated by the gradient projection method with $\alpha^k = 1$ and $\epsilon^k = \epsilon$ for all $k$. 
Assume that for some constant $L > 0$, we have 
\[||\nabla \phi(\textbf{x}) - \nabla \phi(\textbf{y})|| \le L||\textbf{x} - \textbf{y}||, \forall \textbf{x}, \textbf{y} \in \mathcal{X}.\]
Then, if $0 < \epsilon < \frac{2}{L}$, every limit point of $\{\textbf{x}^k\}$ is stationary. 
\end{Proposition}
\end{document}